%% file: ftbfs.tex
\documentclass[11pt]{article}

\usepackage{amsmath}
\usepackage{amssymb}
\usepackage{amsfonts}
\usepackage{amsthm}
\usepackage[utf8]{inputenc}

\usepackage{textcomp}
\usepackage{indentfirst}

\usepackage{abstract}

\usepackage{pbox}
\usepackage{comment}

\usepackage{algorithm}
\usepackage[noend]{algorithmic}

\usepackage[margin=1in]{geometry}

\usepackage{pgfplots}
\pgfplotsset{compat=1.4}
\usetikzlibrary{shapes}
\usetikzlibrary{plotmarks}

\usepackage{framed}
\usepackage{mdframed}
\usepackage{placeins}
\usepackage{subcaption}
\usepackage{afterpage}
\usepackage{hyperref}
\usepackage{courier}
\usepackage{multirow}
\usepackage{mathrsfs}
\usepackage{microtype} 
\usepackage{enumitem}
\usetikzlibrary{patterns}

\allowdisplaybreaks

\newcommand{\dilation}{\mbox{\tt d}}
\newcommand{\congestion}{\mbox{\tt c}}

\newcommand{\eps}{\varepsilon}
\newcommand{\Oish}{\widetilde{O}}

\newcommand{\rr}{\mathbb{R}}

\DeclareMathOperator{\dist}{dist}

\DeclareMathOperator{\srp}{\textsc{subset-rp}}

\newcommand{\poly}{\operatorname{poly}}

\newtheorem{theorem}{Theorem}
\newtheorem{lemma}[theorem]{Lemma}

\newtheorem{corollary}[theorem]{Corollary}
\newtheorem{definition}[theorem]{Definition}

\def\depth{\mbox{\tt depth}}
\def\NSource{\sigma}
\def\Root{\mbox{\tt r}}
\def\Leaf{\mbox{\tt Leaf}}
\def\NLeaf{\mbox{\tt nLeaf}}
\def\NodesIn{\mbox{\tt N}}
\def\LAB{\mbox{\tt Label}}

\usepackage{amsmath}

\newtheorem*{definition*}{Definition}
\newtheorem*{theorem*}{Theorem}

\newtheorem{observation}{Observation}

\newtheorem*{lemma*}{Lemma}

\title{Restorable Shortest Path Tiebreaking for Edge-Faulty Graphs}
\author{Greg Bodwin\thanks{Supported by NSF:AF 2153680}\\University of Michigan EECS\\\texttt{bodwin@umich.edu}
\and
Merav Parter\\Weizmann Institute of Science\\\texttt{merav.parter@weizmann.ac.il}}
\date{}

\begin{document}


\maketitle
\thispagestyle{empty}

\input{abstract}


\tableofcontents

\clearpage

\setcounter{page}{1}

\input{introduction}

\input{tiebreaking}

\input{apps}

\FloatBarrier
\bibliographystyle{plain}
\bibliography{refs}

\appendix
\input{lower-bound}

\end{document}

%% file: abstract.tex
\begin{abstract}
The \emph{restoration lemma} by Afek, Bremler-Barr, Kaplan, Cohen, and Merritt [Dist.\ Comp.\ '02] proves that, in an undirected unweighted graph, any replacement shortest path avoiding a failing edge can be expressed as the concatenation of two original shortest paths.
However, the lemma is \emph{tiebreaking-sensitive}: if one selects a particular canonical shortest path for each node pair, it is no longer guaranteed that one can build replacement paths by concatenating two \emph{selected} shortest paths.
They left as an open problem whether a method of shortest path tiebreaking with this desirable property is generally possible.

We settle this question affirmatively with the first general construction of \emph{restorable tiebreaking schemes}.
We then show applications to various problems in fault-tolerant network design.
These include a faster algorithm for subset replacement paths, more efficient fault-tolerant (exact) distance labeling schemes, fault-tolerant subset distance preservers and $+4$ additive spanners with improved sparsity, and fast distributed algorithms that construct these objects.
For example, an almost immediate corollary of our restorable tiebreaking scheme is the first nontrivial distributed construction of sparse fault-tolerant distance preservers resilient to \emph{three} faults.
\end{abstract}

%% file: introduction.tex
\section{Introduction}

This paper builds on a classic work of Afek, Bremler-Barr, Kaplan, Cohen, and Merritt from 2002, which initiated a theory of \emph{shortest path restoration} in graphs \cite{ABKCM02}.
The premise is that one has a network, represented by a graph $G$, and one has computed its shortest paths and stored them in a routing table.
But then, an edge in the graph breaks, rendering some of the paths unusable.
We want to efficiently \emph{restore} these paths, changing the table to reroute them along a new shortest path between the same endpoints in the surviving graph.
An ideal solution will both avoid recomputing shortest paths from scratch and only require easy-to-implement changes to the routing table.

Motivated by the fact that the multiprotocol label switching (MPLS) allows for efficient concatenation of paths, Afek et al~\cite{ABKCM02} developed the following elegant structure theorem for the problem, called the \emph{restoration lemma}.
All graphs in this discussion are undirected and unweighted.

\begin{theorem} [Restoration Lemma \cite{ABKCM02}]
For any graph $G = (V, E)$, vertices $s,t \in V$, and failing edge $e \in E$, there exists a vertex $x$ and a replacement shortest $s \leadsto t$ path avoiding $e$ that is the concatenation of two original shortest paths $\pi(s, x), \pi(t, x)$ in $G$. 
\end{theorem}

We remark that some versions of this lemma were perhaps implicit in prior work, e.g., \cite{HS01}.
The restoration lemma itself has proved somewhat difficult to apply directly, and most applications of this theory use weaker variants instead (e.g., \cite{bodwin2017preserving, BiloCG0PP18, ChechikCFK17}).
The issue is that the restoration lemma is \emph{tiebreaking-sensitive}, in a sense that we next explain.

To illustrate, let us try a naive attempt at applying the restoration lemma.
One might try to restore a shortest path $\pi(s, t)$ under a failing edge $e$ by searching over all possible midpoint nodes $x$, concatenating the existing shortest paths $\pi(s, x), \pi(t, x)$, and then selecting the replacement $s \leadsto t$ path to be the shortest among all concatenated paths that avoid $e$.
It might seem that the restoration lemma promises that one such choice of midpoint $x$ will yield a valid replacement shortest path.
But this isn't quite right: the restoration lemma only promises that \emph{there exist} two shortest paths of the form $\pi(s, x), \pi(t, x)$ whose concatenation forms a valid replacement path.
Generally there can be many tied-for-shortest $s \leadsto x$ and $t \leadsto x$ paths, and in designing the initial routing table we implicitly broke ties to select just one of them.
The bad case is when for the proper midpoint node $x$, we select the canonical shortest $s \leadsto x$ path to be one that uses the failing edge $f$ (even though the restoration lemma promises that a \emph{different} $s \leadsto x$ shortest path avoids $f$), and thus this restoration-by-concatenation algorithm wrongly discards $x$ as a potential midpoint node.

\begin{figure} [h] \centering
\begin{tikzpicture}
\draw [ultra thick] (0, 0) -- (10, 0);

\node [red] at (5, 0) {\bf \Huge $\times$};
\node [red] at (5, -0.5) {$f$};
\draw [fill=black] (4.5, 0) circle [radius=0.15];
\draw [fill=black] (5.5, 0) circle [radius=0.15];

\draw [ultra thick] (0, 0) to[bend left=20] (5, 2);
\draw [ultra thick] (10, 0) to[bend right=20] (5, 2);

\node [rotate=-18] at (7.5, 1.8) {\small original shortest path};

\draw [line width = 3, blue, ->] plot [smooth] coordinates {(0, 0) (4.5, 0) (5.5, 0) (5, 2)};
\draw [line width = 3, blue, ->] (0, 0) to[bend left=20] (5, 2);

\draw [ultra thick, fill=white] (0, 0) circle [radius=0.3cm];
\node at (0, 0) {$s$};

\draw [ultra thick, fill=white] (10, 0) circle [radius=0.3cm];
\node at (10, 0) {$t$};

\draw [ultra thick, fill=white] (5, 2) circle [radius=0.3cm];

\draw [blue] (2.5, 1.5) -- (3, 0.75) -- (3, 0);

\node [blue, fill=white, rotate=10] at (3, 0.75) {\small tied for shortest};
\node at (5, 2) {$x$};
\end{tikzpicture}
\caption{The restoration lemma is ``tiebreaking-sensitive'' in the sense that there could be several tied-for-shortest paths between $s$ and the midpoint node $x$.  The restoration lemma promises that \emph{one} such path avoids the failing edge $f$, but the $s \leadsto x$ shortest path \emph{selected by the routing table} might still use $f$.}
\end{figure}

So, is it possible to break shortest path ties in such a way that this restoration-by-concatenation method works?
Afek et al.~\cite{ABKCM02} discussed this question extensively, and gave a partial negative resolution: when the input graph is a $4$-cycle, one cannot select \emph{symmetric} shortest paths to enable the method (for completeness, we include a formal proof in Theorem \ref{thm:impossible} in the appendix).
By ``symmetric'' we mean that, for all nodes $s, t$, the selected $s \leadsto t$ and $t \leadsto s$ shortest paths are the same.
However, Afek et al.~\cite{ABKCM02} also point out that the MPLS protocol is inherently asymmetric, and so in principle one can choose different $s \leadsto t$ and $t \leadsto s$ shortest paths.
They left as a central open question whether the restoration lemma can be implemented by an asymmetric tiebreaking scheme (see their remark at the bottom of page 8).
In the meantime, they showed that one can select a larger ``base set'' of $O(mn)$ paths\footnote{More precisely, their base set is generated by first choosing an \emph{arbitrary} set of $n \choose 2$ canonical shortest paths, and then taking every possible path in the base set that consists of a canonical shortest path concatenated with a single extra edge on either end.  From this, we can compute a more precise upper bound on base set size of $\le m(n-1)$.  Correctness of this base set can be proved using Theorem \ref{thm:wtdrestorationintro}.} such that one can restore shortest paths by concatenating two of these paths, and they suggested as an intermediate open question whether their base set size can be improved.
This method has found applications in network design (e.g., \cite{ABKCM02, bodwin2017preserving, BiloCG0PP18, ChechikCFK17}), but these applications tend to pay an overhead associated to the larger base set size.

Despite this, the main result of this paper is a positive resolution of the question left by Afek et al.~\cite{ABKCM02}: we prove that asymmetry is indeed enough to allow restorable tiebreaking in every graph.

\begin{theorem} [Main Result] \label{thm:intromain}
In any graph $G$, one can select a \textbf{single} shortest path for each \textbf{ordered} pair of vertices such that, for any pair of vertices $s, t$ and a failing edge $e$ such that an $s \leadsto t$ path remains in $G \setminus \{e\}$, there is a vertex $x$ and a replacement shortest $s \leadsto t$ path avoiding $e$ that is the concatenation of the selected path $\pi(s, x)$ and the reverse of the selected path $\pi(t, x)$.
\end{theorem}

We emphasize again that this theorem is possible only because we select independent shortest path for each \emph{ordered} pair of vertices, and thus asymmetry is allowed.
The shortest path tiebreaking method used in Theorem \ref{thm:intromain} has a few other desirable properties, outlined in Section \ref{sec:tiebreaking}.
Most importantly it is \emph{consistent}, which implies that the selected paths have the right structure to be encoded in a routing table.
It can also be efficiently computed, using a single call to any APSP algorithm that can handle directed weighted input graphs.

We next overview some of our applications of this theorem in algorithms and network design.
We will not specifically revisit the original application in \cite{ABKCM02} to the MPLS routing protocol, but let us briefly discuss the interaction between Theorem \ref{thm:intromain} and this protocol.
Note that Theorem \ref{thm:intromain} builds a $s \leadsto t$ replacement path by concatenating two paths of the form $\pi(s, x), \pi(t, x)$, which are directed towards a middle vertex $x$.
Since the MPLS protocol can efficiently concatenate \emph{oriented} paths (e.g., of the form $\pi(s, x), \pi(x, t)$), one would likely apply our theorem in this context by carrying two routing tables, one of which encodes our tiebreaking scheme $\pi$ and the other of which encodes its reverse $\overline{\pi}$ (i.e., $\pi(s, t) =: \overline{\pi}(t, s)$).
An $s \leadsto t$ replacement path would be computed by scanning over midpoint nodes $x$, and considering paths formed by concatenating the $s \leadsto x$ shortest path from the first routing table with the $x \leadsto t$ shortest path from the second routing table.
For more details on the MPLS protocol and its relationship to this method of path restoration, we refer to \cite{ABKCM02}.

\subsection{Applications}

\paragraph{Replacement Path Algorithms.}

Our first applications of our restorable shortest path tiebreaking are to computation of replacement paths.
The problem has been extensively studied in the \emph{single-pair} setting, where the input is a graph $G = (V, E)$ and a vertex pair $s, t$, and the goal is to report $\dist_{G \setminus \{e\}}(s, t)$ for every edge $e$ along a shortest $s \leadsto t$ path.
The single-pair setting can be solved in $\Oish(m+n)$ time \cite{HS01, malik1989k}.
Recently, Chechik and Cohen \cite{CC19} introduced the \emph{sourcewise} setting, in which one wants to solve the problem for all pairs in $\{s\} \times V$ simultaneously.
They gave an algorithm with $\Oish\left(m \sqrt{n} + n^2 \right)$ runtime on an $n$-node, $m$-edge graph, and they showed that this runtime is optimal (up to hidden polylog factors) under the Boolean Matrix Multiplication conjecture.
This was subsequently generalized to the $S \times V$ setting by Gupta, Jain, and Modi \cite{GJM20}.

We study the natural \emph{subsetwise} version of the problem, $\srp$, where one is given a graph $G$ and a vertex subset $S$, and the goal is to solve the replacement path problem simultaneously for all pairs in $S \times S$.
We prove:
\begin{theorem}
Given an $n$-vertex, $m$-edge undirected unweighted graph $G$ and $|S| = \sigma$ source vertices, there is a centralized algorithm that solves solves $\srp$ in $O(\sigma m)$ + $\Oish(\sigma^2 n)$ time.
\end{theorem}

We remark that, in the case where most pairs $s, t \in S$ have $\dist_G(s, t) = \Omega(n)$, the latter term in the runtime $\sigma^2 n$ is the time required to write down the output.
So this term is unimprovable, up to the hidden $\log$ factors.
The leading term of $\sigma m$ is likely required for any ``combinatorial'' algorithm to compute single-source shortest paths even in the non-faulty setting.
That is: $\sigma m$ is the time to run BFS search from $\sigma$ sources, and it is widely believed \cite{VW10} that multi-source BFS search is the fastest algorithm to compute $S \times S$ shortest paths in unweighted graphs, except for a class of ``algebraic'' algorithms that rely on fast matrix multiplication as a subroutine and which may be faster when $\sigma, m$ are both large \cite{Seidel95, SZ99}.

\paragraph{Fault-tolerant preservers and additive spanners.}
We next discuss our applications for the efficient constructions of fault-tolerant distance preservers, defined as follows:
\begin{definition}[$S \times T$ $f$-FT Preserver]
A subgraph $H \subseteq G$ is an $S \times T$ $f$-FT preserver if for every $s \in S, t \in T,$ and $F \subseteq E$ of size $|F| \le f$, it holds that
$\dist_{H \setminus F}(s,t)=\dist_{G \setminus F}(s,t).$
\end{definition}
When $T=S$ the object is called a \emph{subset} preserver of $S$, and when $T=V$ (all vertices in the input graph) the object is sometimes called an \emph{FT-BFS structure}, since the $f=0$ case is then solved by a collection of BFS trees. The primary objective for all of these objects is to minimize the size of the preserver, as measured by its number of edges.

For $f=1$, it was shown in \cite{bodwin2017preserving, BCPS20} that one can compute an $S \times S$ $1$-FT preserver with $O(|S|n)$ edges, by properly applying the original version of the restoration lemma by Afek et al.~\cite{ABKCM02}.
Our restorable tiebreaking scheme provides a simple and more general way to convert from $S \times V$ $(f-1)$-FT preservers to $S \times S$ $f$-FT preservers, which also enjoys better construction time, in the centralized and distributed settings. 
For example, for $f=1$ we can compute an $S \times S$ 1-FT preserver simply by taking the union of BFS trees from each source $S$, where each BFS tree is computed using our tiebreaking scheme.
More generally, we get the following bounds:
\begin{theorem}
Given an $n$-vertex graph $G = (V, E)$, a set of source vertices $S \subseteq V$, and a fixed nonnegative integer $f$, there is an $(f+1)$-FT $S \times S$ distance preserver of $G, S$ on
$$O\left( n^{2-1/2^f} |S|^{1/2^f} \right) \mbox{~edges.}$$
\end{theorem}
This bound extends the results of \cite{bodwin2017preserving, BCPS20} to larger $f$; for $\ge 2$ faults (that is, $f \ge 1$ in the above theorem), it was not previously known how to compute preserves of this size.
Moreover, using a standard reduction in the literature, we can use these preservers to build improved fault-tolerant \emph{additive spanners}:
\begin{definition} [FT Additive Spanners]
Given a graph $G = (V, E)$ and a set of source vertices $S \subseteq V$, an $f$-FT $+k$ additive spanner is a subgraph $H$ satisfying
$$\dist_{H \setminus F}(s, t) \leq \dist_{G \setminus F}(s, t)+k$$
for all vertices $s, t \in S$ and sets of $|F| \le f$ failing edges.
\end{definition}

\begin{theorem} \label{thm:ftspanintro}
For any $n$-vertex graph $G = (V, E)$ and nonnegative integer $f$, there is an $(f+1)$-FT $+4$ additive spanner on $O_f\left(n^{1 + 2^f / (2^f + 1)}\right)$ edges.
\end{theorem}

This theorem extends a bound of Bil{\' o}, Grandoni, Gual{\' a}, Leucci, Proietti \cite{BGGLP15}, which establishes single-fault $+4$ spanners on $O(n^{3/2})$ edges; this is exactly the construction one gets by plugging in $f=0$ in the above theorem (the same result is also obtained as a corollary of results in \cite{bodwin2017preserving, BCPS20}).
For $\ge 2$ faults ($f \ge 1$ in the above theorem), $+4$ fault-tolerant spanners of the size given in Theorem \ref{thm:ftspanintro}.
However, there are many other notable constructions of fault-tolerant additive spanners with different $+c$ error bounds; see for example \cite{BGGLP15, Parter17, BCP12, bodwin2017preserving, BCPS15}.

\paragraph{Distributed constructions of fault-tolerant preservers.} Distributed constructions of FT preservers, in the $\mathsf{CONGEST}$  model of distributed computing \cite{Peleg:2000}, attracted attention recently \cite{DinitzK:11,GhaffariP16,DR20,ParterDualDist20}. In the context of exact distance preservers, Ghaffari and Parter \cite{GhaffariP16} presented the first distributed constructions of fault tolerant distance preserving structures. For every $n$-vertex $D$-diameter graph $G=(V,E)$ and a source vertex $s \in V$, they gave an $\widetilde{O}(D)$-round randomized $\mathsf{CONGEST}$ algorithm for computing a $1$-FT $\{s\} \times V$ preserver
with $O(n^{3/2})$ edges. Recently, Parter \cite{ParterDualDist20} extended this construction to $1$-FT $S \times V$ preservers with $\widetilde{O}(\sqrt{|S|}n^{3/2})$ edges and using $\widetilde{O}(D+\sqrt{n |S|})$ rounds. \cite{ParterDualDist20} also presented a distributed construction of source-wise preservers against two \emph{edge}-failures, with $O(|S|^{1/8}\cdot n^{15/8})$ edges and using $\widetilde{O}(D+n^{7/8}|S|^{1/8}+|S|^{5/4}n^{3/4})$ rounds. These constructions immediately yield $+2$ additive spanners resilient to two edge failures with subquadratic number of edges, and sublinear round complexity. To this date, we are still lacking efficient distributed constructions\footnote{By efficient, we mean with subquadratic number of edges and sublinear round complexity.} of $f$-FT preservers (or additive spanners) for $f\geq 3$. In addition, no efficient constructions are known for FT spanners with additive stretch larger than two (which are sparser in terms of number of edges w.r.t the current $+2$ FT-additive spanners). Finally, there are no efficient constructions of subsetwise FT-preservers, e.g., the only distributed construction for $1$-FT $S \times S$  preserver employs the sourcewise construction of $1$-FT $S \times V$ preservers, ending with a subgraph of $O(\sqrt{|S|}n^{3/2})$ edges which is quite far from the state-of-the-art (centralized) bound of $O(|S|n)$ edges. 
In this work, we make a progress along all these directions. Combining the restorable tiebreaking scheme with the work of  \cite{ParterDualDist20} allows us to provide efficient constructions of $f$-FT $S \times S$ preservers for $f \in \{1,2,3\}$ whose size bounds match the state-of-the-art bounds of the centralized constructions.
As a result, we also get the first distributed constructions of $+4$ additive spanners resilient to $f \in \{1,2,3\}$ edge faults. 

\begin{theorem}[Distributed Constructions of Subsetwise FT-Preservers]\label{thm:dist-constructions}
For every $D$-diameter $n$-vertex graph $G=(V,E)$, there exist randomized distributed $\mathsf{CONGEST}$ algorithms for computing:
\begin{itemize}[noitemsep]
\item $1$-FT $S\times S$ preservers with $O(|S|n)$ edges and $\widetilde{O}(D+|S|)$ rounds.
\item $2$-FT $S \times S$ preservers with $O(\sqrt{S}n^{3/2})$ edges and $\widetilde{O}(D+\sqrt{|S|n})$ rounds.
\item $3$-FT $S \times S$ preservers with $O(|S|^{1/8}\cdot n^{15/8})$ edges and $\widetilde{O}(D+n^{7/8}|S|^{1/8}+|S|^{5/4}n^{3/4})$ rounds.
\end{itemize}
\end{theorem}
Using the $f$-FT $S\times S$ preservers for $f\in \{1,2,3\}$ for a subset $S$ of size $\sigma \in \{\sqrt{n}, n^{1/3}, n^{1/9}\}$ respectively, we get the first distributed constructions of $f$-FT $+4$ additive spanners. 
\begin{corollary}[Distributed Constructions of FT-Additive Spanners]\label{cor:dist-add-constructions}
For every $D$-diameter $n$-vertex graph $G=(V,E)$, there exist randomized distributed algorithms for computing:
\begin{itemize}[noitemsep]
\item $1$-FT $+4$ additive spanners with $\widetilde{O}(n^{3/2})$ edges and $\widetilde{O}(D+\sqrt{n})$ rounds.
\item $2$-FT $+4$ additive spanners with $\widetilde{O}(n^{5/3})$ edges and $\widetilde{O}(D+n^{5/6})$ rounds.
\item $3$-FT $+4$ additive spanners preservers with $\widetilde{O}(n^{17/9})$ edges and $\widetilde{O}(D+n^{8/9})$ rounds.
\end{itemize}
\end{corollary}

One can also remove the $\log$ factors in these spanner sizes in exchange for an edge bound that holds in expectation, instead of with high probability.

\paragraph{Fault-Tolerant Exact Distance Labeling}

A \emph{distance labeling scheme} is a way to assign short bitstring labels to each vertex of a graph $G$ such that $\dist(s, t)$ can be recovered by inspecting only the labels associated with $s$ and $t$ (and no other information about $G$) \cite{GPPR04,FKMS05,SK85}. In an \emph{$f$-FT} distance labeling scheme, the labels are assigned to both the vertices and the edges of the graph, such that for any set of $|F| \le f$ failing edges, we can even recover $\dist_{G \setminus F}(s, t)$ by inspecting only the labels of $s, t$ and the edge set $F$.
These are sometimes called \emph{forbidden-set labels} \cite{CT07,CGKT07}, and they have been extensively studied in specific graph families, especially due to applications in routing \cite{ACG12, ACGP16}.
The prior work mainly focused on connectivity labels and \emph{approximate} distance labels. In those works, the labels were given also to the edges, and one can inspect the labels of failing edges as well.

In our setting we consider \emph{exact} distance labels. Interestingly, our approach will not need to use edge labels; that is, it recovers $\dist_{G \setminus F}(s, t)$ only from the labels of $s, t$ and a description of the edge set $F$.
Since one can always provide the entire graph description as part of the label, our main objective is in providing FT exact distance labels of \emph{subquadratic} length, 
To the best of our knowledge, the only prior labeling scheme for recovering exact distances under faults was given by Bil{\' o} et al.~\cite{BCGLPP18}. They showed that given a source vertex $s$, one can recover distances in $\{s\} \times V$ under one failing edge using labels of size $O(n^{1/2})$. For the all-pairs setting, this would extend to label sizes of $O(n^{3/2})$ bits.

We prove:
\begin{theorem}[Subquadratic FT labels for Exact Distances]
For any fixed nonnegative integer $f \geq 0$, and $n$-vertex unweighted undirected graph, there is an $(f+1)$-FT distance labeling scheme that assigns each vertex a label of
$$O\left(n^{2 - 1/2^f} \log n\right) \text{ bits}.$$
\end{theorem}

For $f=0$ our vertex labels have size $\Oish(n)$, improving over $\Oish(n^{3/2})$ from \cite{BCGLPP18}.
Our size is near-optimal, for $f=0$, in the sense that $\Omega(n)$ label sizes are required even for non-faulty exact distance labeling. This is from a simple information-theoretic lower bound: one can recover the graph from the labeling, and there are $2^{\Theta(n^2)}$ $n$-vertex graphs, so $\Omega(n^2)$ bits are needed in total.

Finally, we remark that FT labels for exact distances are also closely related to distance sensitivity oracles: these are global and centralized data structures that reports $s$-$t$ distances in $G \setminus F$ efficiently. For any $f=O(\log n/\log\log n)$, Weimann and Yuster \cite{weimann2010replacement,weimann2013replacement} provided a construction of distance sensitivity oracles using subcubic space and subquadratic query time. The state-of-the-art bounds for this setting are given by van den Brand and Saranurak \cite{van2019sensitive}. It is unclear, however, how to balance the information of these global succinct data-structures among the $n$ vertices, in the form of distributed labels.

\subsection{Other Graph Settings}

The results of this paper do not extend to directed and/or weighted graphs; we use both undirectedness and unweightedness in the proof of our main theorem.
Indeed, it was noted in \cite{ABKCM02, bodwin2017preserving} that the restoration lemma itself is not still generally true for graphs that are weighted and/or directed, so there is not much hope for a direct extension of Theorem \ref{thm:intromain} to these settings.
However, we will briefly discuss two extensions of this theory to other graph settings that appear in prior work.
First, the original work of Afek et al.~\cite{ABKCM02} included the following version of the restoration lemma for weighted graphs:
\begin{theorem} [Weighted Restoration Lemma \cite{ABKCM02}] \label{thm:wtdrestorationintro}
For any undirected graph $G = (V, E, w)$ with positive edge weights, vertices $s, t \in V$, and failing edge $e \in E$, there exists an edge $(u, v)$ such that for \textbf{any} shortest paths $\pi(s, u), \pi(v, t)$, the path $\pi(s, u) \circ (u, v) \circ \pi(v, t)$ is a replacement $s \leadsto t$ shortest path avoiding $e$.
\end{theorem}

The weighted restoration lemma gives a weaker structural property than the original restoration lemma, since it includes a middle edge between the two concatenated shortest paths.
But, it is not tiebreaking-sensitive and it extends to weighted graphs, which make it useful in some settings (in particular to generate small base sets without reference to a particular tiebreaking scheme, as discussed above).
The second extension to mention is that the weighted and unweighted restoration lemmas both extend to DAGs, and so many of their applications extend to DAGs as well \cite{ABKCM02, bodwin2017preserving}.
It seems very plausible that our main result admits some kind of extension to unweighted DAGs, but we leave the appropriate formulation and proof as a direction for future work.

%% file: tiebreaking.tex
\section{Replacement Path Tiebreaking Schemes \label{sec:tiebreaking}}

In this section we formally introduce the framework of \emph{tiebreaking schemes} in network design, and we extend it into the setting where graph edges can fail.
The following objects are studied in prior work on non-faulty graphs:
\begin{definition} [Shortest Path Tiebreaking Schemes]
In a graph $G$, a \emph{shortest path tiebreaking scheme} $\pi$ is a function from vertex pairs $s, t$ to one particular shortest $s \leadsto t$ path $\pi(s, t)$ in $G$ (or $\pi(s, t) := \emptyset$ if no $s \leadsto t$ path exists).
\end{definition}

It is often useful to enforce coordination between the choices of shortest paths.
Two basic kinds of coordination are:
\begin{definition} [Symmetry]
A tiebreaking scheme $\pi$ is \emph{symmetric} if for all vertex pairs $s, t$, we have $\pi(s, t) = \pi(t, s)$ (when $\pi(s, t), \pi(t, s)$ are viewed as undirected paths).
\end{definition}

\begin{definition} [Consistency]
A tiebreaking scheme $\pi$ is \emph{consistent} if, for all vertices $s, t, u, v$, if $u$ precedes $v$ in $\pi(s, t)$, then $\pi(u, v)$ is a contiguous subpath of $\pi(s, t)$.
\end{definition}

Consistency is important for many reasons; it is worth explicitly pointing out the following two:
\begin{itemize}
\item As is well known, in any graph $G = (V, E)$ one can find a subtree that preserves all $\{s\} \times V$ distances.
Consistency gives a natural converse to this statement: if one selects shortest paths using a consistent tiebreaking scheme and then overlays the $\{s\} \times V$ shortest paths, one gets a tree.

\item It is standard to encode the shortest paths of a graph $G$ in a \emph{routing table} -- that is, a matrix indexed by the vertices of $G$ whose $(i, j)^{th}$ entry holds the vertex ID of the next hop on an $i \leadsto j$ shortest path.
For example, the standard implementation of the Floyd-Warshall shortest path algorithm outputs shortest paths via a routing table of this form.
Many routing tables in practice work similarly; for example, routing tables for the Internet typically encode only the next hop on the way to the destination.
Once again, consistency gives a natural converse: if one selects shortest paths in a graph using consistent tiebreaking, then it is possible to encode these paths in a routing table.
\end{itemize}

Since our goal is to study tiebreaking under edge faults, we introduce the following extended definition:
\begin{definition} [$f$-Replacement Path Tiebreaking Schemes]
In a graph $G = (V, E)$, an \emph{$f$-replacement path tiebreaking scheme ($f$-RPTS)} is a function of the form $\pi(s, t \mid F)$, where $s, t \in V$ and $F \subseteq E, |F| \le f$.
The requirement is that, for any fixed set of failing edges $F$ of size $|F| \le f$, the two-parameter function $\pi(\cdot, \cdot \mid F)$ is a shortest path tiebreaking scheme in the graph $G \setminus F$.
\end{definition}

We will say that an RPTS $\pi$ is symmetric or consistent if, for any given set $F$ of size $|F| \le f$, the tiebreaking scheme $\pi(\cdot, \cdot \mid F)$ is symmetric or consistent over the graph $G \setminus F$.
We also introduce the following natural property, which says that selected shortest paths do not change unless this is forced by a new fault:
\begin{definition} [Stability]
An $f$-RPTS $\pi$ is \emph{stable} if, for all $s, t, F$ of size $|F| \le f-1$ and edge $f \notin \pi(s, t \mid F)$, we have $\pi(s, t \mid F) = \pi(s, t \mid F \cup \{f\})$.
\end{definition}

Our main result is expressed formally using an additional coordination property of RPTSes that we call \emph{restorability}.
We will discuss that in the following section.

\section{Restorable Tiebreaking \label{sec:restorable}}

The main result in this paper concerns the following new coordination property:
\begin{definition} [$f$-Restorable Tiebreaking]
An RPTS $\pi$ is \emph{$f$-restorable} if, for all vertices $s, t$ and nonempty edge fault sets $F$ of size $|F| \le f$, there exists a vertex $x$ and a proper fault subset $F' \subsetneq F$ such that the concatenation of the paths $\pi(s, x \mid F'), \pi(t, x \mid F')$ forms an $s \leadsto t$ replacement path avoiding $F$.
(Note: it is not required that this concatenated path is specifically equal to $\pi(s, t \mid F)$, just that it is one of the possible replacement paths.)
\end{definition}

The motivation for this definition comes from restoration lemma by Afek et al~\cite{ABKCM02}, discussed above.
It is not obvious at this point that any or all graphs should admit a restorable RPTS.
That is the subject of our main result.

\subsection{Antisymmetric Tiebreaking Weight Functions and Restorability}

We will analyze a class of RPTSes generated by the following method.
Let $G = (V, E)$ be the undirected unweighted input graph.
Convert $G$ to a directed graph by replacing each undirected edge $(u, v) \in E$ with both directed edges $\{(u, v), (v, u)\}$.
For this symmetric directed graph, we define:
\begin{definition} [Antisymmetric Tiebreaking Weight Function] \label{def:atw}
A function $r : E \to \rr$ is an antisymmetric $f$-fault tiebreaking weight (ATW) function for $G$ if it satisfies:
\begin{itemize}
\item (Antisymmetric) $r(u, v) = -r(v, u)$ for all $(u, v) \in E$,

\item ($f$-Fault Tiebreaking) Let $G^*$ be the directed weighted graph obtained by setting the weight of each edge in $G$ to be $w(u, v) := 1 + r(u, v)$.
The requirement is that for any edge subset $F$ of size $|F| \le f$, the graph $G^* \setminus F$ has a unique shortest path for each node pair, and moreover these unique shortest paths are each a shortest path in the graph $G \setminus F$.
\end{itemize}
\end{definition}

Any antisymmetric $f$-fault tiebreaking weight function $r$ naturally generates an $f$-RPTS for $G$, in which $\pi(s, t \mid F)$ is the unique shortest $s \leadsto t$ path in the graph $G^* \setminus F$.
We next prove that any $f$-RPTS generated in this way grants $f$-restorability.
Afterwards, we will discuss issues related to existence and computation of antisymmetric $f$-fault tiebreaking weight functions.

\begin{theorem} \label{thm:restorabletiebreaking}
Any $f$-RPTS $\pi$ generated by an antisymmetric $f$-fault tiebreaking weight function $r$ is simultaneously stable, consistent, and $f$-restorable.
\end{theorem}
\begin{proof}
Consistency and stability follow immediately from the fact that the paths selected by $\pi$ under fault set $F$ are unique shortest paths in the graph $G^* \setminus F$.
The rest of this proof establishes that $\pi$ is $f$-restorable.
We notice that it suffices to consider only the special case $f=1$, for the following reason.
Suppose we are analyzing $f$-restorability, and we consider an arbitrary set $F$ of $|F| \le f$ failing edges.
We can select an arbitrary subset $F' \subseteq F$ of all but one failing edge.
We can then view $G \setminus F'$ as the input graph, and we can view $\pi$ as a $1$-RPTS over this input graph.
If we can prove that $\pi$ is $1$-restorable on $G \setminus F'$, this implies in turn that $\pi$ is $f$-restorable over $G$.

So, the following proof assumes $f = 1$.
Let $s, t$ be vertices and let $(u, v)$ be the one failing edge.
We may assume without loss of generality that $(u, v) \in \pi(s, t)$, with that orientation, since otherwise by stability we have $\pi(s, t) = \pi(s, t \mid (u, v))$ and so claim is immediate by choice of (say) $x=t$.
Let $x$ be the last vertex along $\pi(s, t \mid (u, v))$ such that $\pi(s, x)$ avoids $(u, v)$, and let $y$ be the vertex immediately after $x$ along $\pi(s, t \mid (u, v))$.
Hence $(u, v) \in \pi(s, y)$.
These definitions are all recapped in the following diagram.

\begin{center}
\begin{tikzpicture}
\draw [fill=black] (0, 0) circle [radius=0.15];
\draw [fill=black] (7, 0) circle [radius=0.15];
\node at (0, -0.5) {$s$};
\node at (7, -0.5) {$t$};

\draw [fill=black] (3, 0) circle [radius=0.15];
\draw [fill=black] (4, 0) circle [radius=0.15];
\node at (3.5, 0) {\Huge $\mathbf \times$};
\node at (3, -0.5) {$u$};
\node at (4, -0.5) {$v$};

\draw [fill=black] (3, 1) circle [radius=0.15];
\draw [fill=black] (4, 1) circle [radius=0.15];
\node at (3, 1.5) {$x$};
\node at (4, 1.5) {$y$};
\draw [dashed] plot coordinates {(0, 0) (3, 1) (4, 1) (7, 0)};
\node [rotate=-21] at (5.8, 0.8) {$\pi(s, t \mid f)$};
\draw (0, 0) -- (3, 1);

\draw plot [smooth] coordinates {(0, 0) (3, 0) (4, 0) (4, 1)};
\node at (1.5, -0.4) {$\pi(s, y)$};
\node [rotate=21] at (1.2, 0.8) {$\pi(s, x)$};

\end{tikzpicture}
\end{center}

Our goal is now to argue that $\pi(t, x)$ does not use the edge $(u, v)$, and thus $\pi(s, x) \cup \pi(t, x)$ forms a replacement path avoiding $(u, v)$.
Since we have assumed that $\dist_G(u, t) > \dist_G(v, t)$, if $(v, u) \in \pi(t, x)$ then it appears with that particular orientation, $v$ preceding $u$.

In the following we will write $\dist^*(\cdot, \cdot)$ for the distance function in the directed reweighted graph $G^*$ (so $\dist^*$ is not necessarily integral, and it is asymmetric in its two parameters).
Recall that $\pi(s, y)$ includes the edge $(u, v)$, and hence the $u \leadsto y$ path through $v$ is shorter than the alternate $u \leadsto y$ path through $x$.
We thus have the inequality:
$$\dist^*(u, v) + \dist^*(v, y) < \dist^*(u, x) + \dist^*(x, y).$$
Since $(u, v), (x, y)$ are single edges we can write
$$(1 + r(u, v)) + \dist^*(v, y) < \dist^*(u, x) + (1 + r(x, y)),$$
Rearranging and using antisymmetry of $r$, we get
$$(1 + r(y, x)) + \dist^*(v, y) < \dist^*(u, x) + (1 + r(v, u)),$$
and so
$$\dist^*(v, y) + \dist^*(y, x) < \dist^*(v, u) + \dist^*(u, x).$$
This inequality says that the $v \leadsto x$ path that passes through $y$ is shorter in the reweighted graph than the one that passes through $u$.
So $(v, u) \notin \pi(v, x)$.
By consistency and the previously-mentioned fact that $\dist(v, t) < \dist(u, t)$, this also implies that $(v, u) \notin \pi(t, x)$, as desired.
\end{proof}

To complete our main result, we still need to prove existence of antisymmetric tiebreaking weight functions.
\begin{theorem} \label{thm:rwtexists}
Every undirected unweighted graph $G$ admits an antisymmetric $f$-fault tiebreaking weight function $r$ (for any $f$).
\end{theorem}
\begin{proof}
The simplest way to generate $r$ is randomly.
Let $n$ be the number of nodes in $G$, and let $\eps < 1/(2n)$.
For each edge $(u, v) \in G$, set $r(u, v)$ to a uniform random real number in the interval $[-\eps, \eps]$, and set $r(v, u) := -r(u, v)$.
Antisymmetry of $r$ is immediate.
We then need to argue that $r$ acts as tiebreaking edge weights (with probability $1$).

Fix an edge subset $F$ and nodes $s, t$.
Consider two tied-for-shortest $s \leadsto t$ paths $q_1, q_2$ in the graph $G \setminus F$.
The probability that the lengths of $q_1, q_2$ remain tied in $G^*$ is $0$.
To see this, consider an edge $e \in q_1 \setminus q_2$.
If we fix the value of $r$ on all other edges in $(q_1 \cup q_2) \setminus e$, then there is at most one value of $r(e)$ that would cause the lengths of $q_1, q_2$ to tie.
The probability we set $r(e)$ to exactly this value in the interval $[-\eps, \eps]$ is $0$.
Thus, with probability $1$, among the set of shortest $s \leadsto t$ paths in $G \setminus F$ there will be a unique one of minimum length in $G^*$.

Finally, we need to show that no non-shortest path in $G$ becomes a shortest path in $G^*$.
Let $q_1$ again be a shortest $s \leadsto t$ path in $G \setminus F$, and let $q_2$ be a non-shortest simple $s \leadsto t$ path in $G \setminus F$.
Since $G \setminus F$ is unweighted, the length of $q_2$ is at least $1$ more than the length of $q_1$.
Since $\eps < 1/(2n)$, and since $q_1, q_2$ each contain at most $n-1$ edges, the length of $q_1$ increases over the reweighting in $G^*$ by $<1/2$, and the length of $q_2$ decreases over the reweighting in $G^*$ by $<1/2$.
Thus $q_1$ remains strictly shorter than $q_2$ in $G^*$.
\end{proof}

\subsection{Bit Complexity and Determinism \label{sec:detbit}}

Before continuing, we address two possible shortcomings in the proof of Theorem \ref{thm:rwtexists}.
The first concern is that this proof operates in the real-RAM model: we allow ourselves to perturb edge weights by arbitrary real numbers in the interval $[-\eps, \eps]$.
Practical implementations might need to care about the \emph{bit complexity} of $r$; that is, they might pay a time/space penalty if there is a significant space overhead to representing the edge weights in the graph $G^*$.

Fortunately, there is a bit-efficient solution to this problem, based on an application of the \emph{isolation lemma} of Mulmuley, Vazirani, and Vazirani.
We will state a slightly special case of the isolation lemma here (for paths in graphs):
\begin{theorem} [Isolation Lemma \cite{MVV87}]
Let $G = (V, E)$ be a graph and let $\Pi$ be a set of paths in $G$.
Suppose we choose an integer weight for each edge in $G$ uniformly at random from the range $\{1, \dots, W\}$.
Then, with probability at least $1 - |E|/W$, there is a unique path $\pi \in \Pi$ that has minimum length among the paths in $\Pi$.
\end{theorem}

The surprise in the isolation lemma is that there is no dependence on the number of paths $|\Pi|$: there can even be exponentially many paths in $\Pi$, and yet we will still have a shortest one with good probability.
This makes it very helpful for tiebreaking applications, like the following.

\begin{corollary}
In any $n$-node undirected unweighted graph $G = (V, E)$ and integer $f \ge 1$, there is a randomized polynmoial time that returns an antisymmetric $f$-fault tiebreaking weight function $r$ for which each value $r(u, v)$ can be represented in $O(f \log n)$ bits.
\end{corollary}
\begin{proof}
Set $W := n^{f+4+c}$, where $c$ is a positive integer constant that we will choose later.
Set the values of the weight function $r(u, v)$ by selecting a uniform-random number from among the $2W+1$ numbers in set
$$\left\{\frac{i}{W} \cdot \frac{1}{2n} \ \mid \ i \in \{-W, -W+1 \dots, W-1, W\} \right\},$$
and set $r(v, u) := -r(u, v)$ to enforce antisymmetry.
Note that we need to encode one of $n^{f+4+c}$ values for each edge weight, which requires $\log (n^{f+4+c}) = O(f \log n)$ bits per edge.
In the following, let $G^*$ be the directed reweighted version of $G$ according to weight function $r$.

Now fix nodes $s, t$ and set of edges $F$ with $|F| \le f$, let $\Pi$ be the set of shortest $s \leadsto t$ paths in the graph $G \setminus F$.
Exactly as in Theorem \ref{thm:rwtexists}, since path lengths change by $<1/(2n)$ over the reweighting, no non-shortest path in $G \setminus F$ can become a shortest path in the graph $G^* \setminus F$ reweighted by $r$.
Meanwhile, by the isolation lemma, with probability at least 
$$1 - \frac{|E|}{n^{f+4+c}} > 1 - \frac{1}{n^{f+2+c}},$$
there is a unique path in $\Pi$ whose weight decreases the most over the reweighting to $G^*$, which is thus the unique shortest $s \leadsto t$ path in $G^* \setminus F$.
By a union bound over the $\le n^{f+2}$ possible choices of $s, t, F$, there is $\ge 1 - \frac{1}{n^c}$ probability that for \emph{every} possible choice of $s, t, F$, there is a unique shortest $s \leadsto t$ path in the graph $G^* \setminus F$.
Thus, with high probability (controlled by the choice of $c$), $r$ is $f$-tiebreaking.
\end{proof}

The second possible shortcoming of this approach is that it is \emph{randomized}; it is natural to ask whether one can achieve an antisymmetric weight function \emph{deterministically}.
There one easy way to do so, using a slight tweak on a folklore method:
\begin{theorem}
There is a deterministic polynomial time algorithm that, given an $n$-node graph $G = (V, E)$, computes antisymmetric $f$-tiebreaking edge weights (for any $f$) using $O(|E|)$ bits per edge.
\end{theorem}
\begin{proof}
Assume without loss of generality that $G$ is connected (otherwise one can compute tiebreaking edge weights over each connected component individually).
We define our weight function $r$ as follows.
Arbitrarily assign the edges ID numbers $i \in \{1, \dots, |E|\}$.
Then, assign edge weights
$$r(u, v) := \text{sign}(u-v) \cdot C^{-i} \cdot \frac{1}{2n}.$$
where $C$ is some large enough absolute constant.
Antisymmetry is immediate, and we note that this expression is represented in $O(|E| + \log n) = O(|E|)$ bits.
To show that $r$ is indeed an $f$-tiebreaking function, let $G^*$ be the directed weighted graph associated to $r$, let $F$ be an arbitrary set of edge faults.
\begin{itemize}
\item First, we argue that no non-shortest path in $G \setminus F$ can become a shortest path in $G^* \setminus F$.
This part is essentially identical to Theorem \ref{thm:rwtexists}:
since $G$ is unweighted, the initial length of any non-shortest path is at least $1$ less than the length of a shortest path.
In $G^* \setminus F$, each edge weight changes by $\le 1/(2n)$, and thus the length of any simple path $\pi$ changes by $<1/2$.
It follows that no non-shortest path in $G \setminus F$ can become a shortest path in $G^* \setminus F$.

\item Next, we argue that $G^* \setminus F$ has \emph{unique} shortest paths.
Fix nodes $s, t$ and let $\pi_1, \pi_2$ be distinct $s \leadsto t$ (not-necessarily-shortest) paths.
Let $(u, v)$ be the edge in $\pi_1 \setminus \pi_2$ with smallest ID, and suppose without loss of generality that $\text{sign}(u, v) = 1$.
Then, since edge weights decrease geometrically, the weight increase in $(u, v)$ must be strictly larger than the total weight change over other edges in the symmetric difference of $\pi_1, \pi_2$.
Thus $\pi_2$ is shorter than $\pi_1$.
It follows that there cannot be two equally-shortest $s \leadsto t$ paths in $G^* \setminus F$. \qedhere
\end{itemize}
\end{proof}

One can ask whether there is a construction of antisymmetric $f$-fault tiebreaking weight functions that is simultaneously deterministic and which achieves reasonable bit complexity, perhaps closer to the $O(f \log n)$ bound achieved by the isolation lemma.
Proving or refuting such a construction for general graphs is an interesting open problem.
In fact, it is even open to achieve an algorithm that deterministically computes tiebreaking edge weights with bit complexity $O(\log n)$ (or perhaps a bit worse), even in the \emph{non-faulty} $f=0$ setting, and even if we do not require antisymmetry.

That said, we comment that the reweighting used in the previous theorem will only assign $O(|E|)$ different possible edge weights.
Thus, if we give up the standard representation of numbers, we could remap these edge weights and represent them using only $O(\log n)$ bits per edge.

%% file: apps.tex
\section{Applications}

\subsection{Fault-Tolerant Preservers and Tiebreaking Schemes}

We will begin by recapping some prior work on fault-tolerant $S \times V$ preservers in light of our RPTS framework (which will be used centrally in our applications).
We supply a new auxiliary theorem to fill in a gap in the literature highlighted by this framework.

Fault-tolerant $S \times V$ preservers were introduced by Parter and Peleg \cite{PP14}, initially called \emph{FT-BFS structures} (since the $f=0$ case is solved by BFS trees).
We mentioned earlier that, for a tiebreaking scheme $\pi$, one gets a valid BFS tree by overlaying the $\{s\} \times V$ shortest paths selected by $\pi$ if $\pi$ is consistent.
To generalize this basic fact, it is natural to ask what properties of an $f$-RPTS yield an $f$-FT $S \times V$ preserver of optimal size, when the structure is formed by overlaying all the replacement paths selected by $\pi$.
For $f=1$, this question was settled by Parter and Peleg \cite{PP14}:
\begin{theorem} [\cite{PP14}] \label{thm:oneftbfs} ~
For any $n$-vertex graph $G = (V, E)$, set of source vertices $S \subseteq V$, and consistent stable $1$-RPTS $\pi$, the $1$-FT $S \times V$ preserver formed by overlaying all $S \times V$ replacement paths selected by $\pi$ under $\le 1$ failing edge has $O\left(n^{3/2} |S|^{1/2}\right)$ edges. This bound is existentially tight. 
%
\end{theorem}
The extension to $f=2$ was established by Parter \cite{parter2015dual} and Gupta and Khan \cite{GK19}:
\begin{theorem}[\cite{parter2015dual,GK19}] \label{thm:2ftbfs}~
For any $n$-vertex graph $G = (V, E)$ and $S \subseteq V$, there is a $2$-FT $S \times V$ preserver with
$$O\left(n^{5/3} |S|^{1/3}\right) \mbox{~edges}.$$
This bound is also existentially tight.
\end{theorem}
In contrast to Theorem \ref{thm:oneftbfs} which works with any stable and consistent $1$-RPTS, the upper bound part of Theorem \ref{thm:2ftbfs} uses additional properties beyond stability and consistency (in \cite{GK19}, called ``preferred'' paths).
For general fixed $f$, Parter \cite{parter2015dual} gave examples of graphs where any preserver $H$ has
$$|E(H)| = \Omega\left(n^{2 - \frac{1}{f+1}} |S|^{\frac{1}{f+1}} \right).$$
It is a major open question in the area to prove or refute tightness of this lower bound for any $f\geq 3$.
The only general upper bound is due to Bodwin, Grandoni, Parter, and Vassilevska-Williams \cite{bodwin2017preserving}, who gave a randomized algorithm that constructed preservers of size
$$|E(H)| = \Oish\left(n^{2 - 1/2^f} |S|^{1/2^f}\right)$$
in $\poly(n, f)$ time.
In fact, as we explain next, a reframing of this argument shows that the $\log$ factors are shaved by the algorithm that simply overlays the paths selected by a consistent stable RPTS.
(Note: naively, the runtime to compute all these paths is $n^{O(f)}$, hence much worse than the one obtained in \cite{bodwin2017preserving} and polynomial only if we treat $f$ as fixed.)

\begin{theorem} [\cite{bodwin2017preserving}] \label{thm:expftbfs}
For any $n$-vertex graph $G = (V, E)$, $S \subseteq V$, fixed nonnegative integer $f$, and consistent stable $f$-RPTS $\pi$, the $f$-FT $S \times V$ preserver formed by overlaying all $S \times V$ replacement paths selected by $\pi$ under $\le f$ failing edges has
$$O\left(n^{2 - 1/2^f} |S|^{1/2^f}\right) \mbox{~edges.}$$
\end{theorem}
\begin{proof} [Proof Sketch]
At a very high level, the original algorithm in \cite{bodwin2017preserving} can be viewed as follows.
Let $\pi$ be a consistent stable RPTS (stability is used throughout the argument, and consistency is mostly used to enable the ``last-edge observation'' -- we refer to \cite{bodwin2017preserving} for details on this observation).
Consider all the replacement paths $\{\pi(s, v \mid F)\}$ for $s \in S, v \in V$, and fault sets $F$, and consider the ``final detours'' of these paths -- the exact definition of a final detour is not important in this exposition, but the final detour is a suffix of the relevant path.
We then separately consider paths with a ``short'' final detour, of length $\le \ell$ for some appropriate parameter $\ell$, and ``long'' final detours of length $>\ell$.

It turns out that we can add all edges in short final detours to the preserver without much trouble.
To handle the edges in long final detours, we randomly sample a node subset $R$ of size $|R| = \Theta((n \log n) / \ell)$, and we argue that it hits each long detour with high probability.
The subsample $R$ is then used to inform some additional edges in the preserver, which cover the paths with long detours.

Let us imagine that we run the same algorithm, but we only sample a node subset of size $|R| = \Theta(n/\ell)$.
This removes the $\log$ factors from the size of the output preserver, but in exchange, any given path with a long detour is covered in the preserver with only \emph{constant} probability.
In other words: we have a randomized construction that gives an incomplete ``preserver'' $H'$ with only
$$|E(H')| = O\left(n^{2 - 1/2^f} |S|^{1/2^f}\right)$$
edges, which includes any given replacement path $\pi(s, v \mid F)$ with constant probability or higher.
Thus, if we instead consider our original algorithm for the preserver $H$, which includes every replacement path $\pi(s, v \mid F)$ with probability $1$, then $H$ will have more edges than $H'$ by only a constant factor.
So we have
$$|E(H)| = O\left(n^{2 - 1/2^f} |S|^{1/2^f}\right)$$
as well, proving the theorem.
\end{proof}

Plugging in $f=1$ to the bound in the previous theorem recovers the result of Parter and Peleg \cite{PP14}, but this bound is non-optimal already for $f=2$. It is natural to ask whether the analysis of consistent stable RPTSes can be improved, or whether additional tiebreaking properties are necessary.
As a new auxiliary result, we show the latter: consistent and stable tiebreaking schemes cannot improve further on this bound, and in particular these properties alone are non-optimal at least for $f=2$. 

\begin{theorem} \label{thm:cslowerbound}
For any nonnegative integers $f,\sigma$, there are $n$-vertex unweighted (directed or undirected) graphs $G = (V, E)$, vertex subsets $S \subseteq V$ of size $|S| = \sigma$, and $f$-RPTSes that are consistent, stable, and (for undirected graphs) symmetric that give rise to $f$-FT $S \times V$ preservers on $\Omega(n^{2 - 1/2^f} \sigma^{1/2^f})$ edges.
\end{theorem}

The proof of this theorem is given in Appendix \ref{app:conslb}.
We also note that the lower bound of Theorem \ref{thm:cslowerbound} holds under a careful selection of one particular ``bad" tiebreaking scheme, which happens to be consistent, stable, and symmetric.
This lower bound breaks, for example, for many more specific classes of tiebreaking: for example, if one uses small perturbations to edge weights to break the ties. It is therefore intriguing to ask whether it is possible to provide $f$-FT $S \times V$ preservers with optimal edge bounds using random edge perturbations to break the replacement paths ties.  The dual failure case of $f=2$ should serve here as a convenient starting point. It would be very interesting, for example, to see if one can replace the involved tiebreaking scheme of \cite{parter2015dual,GK19} (i.e., of preferred paths) using random edge perturbations while matching the same asymptotic edge bound (or alternatively, to prove otherwise), especially because these random edge perturbations enable restorable tiebreaking.

\subsection{Subset Replacement Path Algorithm}

In the $\srp$ problem, the input is a graph $G = (V, E)$ and a set of source vertices $S$, and the output is: for every pair of vertices $s, t \in S$ and every failing edge $e \in E$, report $\dist_{G \setminus \{e\}}(s, t)$.
We next present our algorithm for $\srp$.
We will use the following algorithm in prior work, for the single-pair setting:

\begin{theorem} [\cite{HS01}] \label{thm:singlepairalg}
When $|S|=2$, there is an algorithm that solves $\srp(G, S)$ in time $\Oish(m + n)$.
\end{theorem}
\begin{proof} [Proof Sketch]
Although we will use this theorem as a black box, we sketch its proof anyways, for completeness.
Let $S = \{s, t\}$, and suppose we perturb edge weights in $G$ to make shortest paths unique.
Compute shortest path trees $T_s, T_t$ leaving $s, t$ respectively.
Applying the weighted restoration lemma (Theorem \ref{thm:wtdrestorationintro}), every edge $(u, v)$ \emph{uniquely} defines a candidate $s \leadsto t$ replacement path, by concatenating the $s \leadsto u$ path in $T_s$ to the edge $(u, v)$ to the $v \leadsto t$ replacement path in $T_t$.
Notice that, after computing distances from $s$ and $t$, we can determine the length of this candidate path in constant time.
We may therefore quickly sort the edges $\{(u, v)\}$ by length of their associated candidate replacement path.

We then consider the edges along the true shortest path $\pi(s, t)$; initially all these edges are \emph{unlabeled}.
Stepping through each edge $(u, v)$ in our list in order, we look at its associated candidate replacement path $q \setminus (u, v)$.
For every still-unlabeled edge $e \in \pi(s, t)$ that is \emph{not} included in $q$, we label that edge with the length of $q$.
That is: since $q$ is the shortest among all candidate replacement paths that avoids $e$, it must witness the $s \leadsto t$ distance in the event that $e$ fails.

It is nontrivial to quickly determine the edge(s) along $\pi(s, t)$ that are both unlabeled and not contained in $q$.
This requires careful encoding of the shortest path trees $T_s, T_t$ in certain data structures.
We will not sketch these details; we refer to \cite{HS01} for more information.
\end{proof}

We then use Algorithm \ref{alg:srp} to solve $\srp$ in the general setting.

\begin{algorithm}
\caption{$\srp$ Algorithm \label{alg:srp}}
\label{ALG:basic}
\begin{algorithmic}[1]
\REQUIRE Undirected unweighted graph $G = (V, E)$ on $n$ vertices, vertex subset $S \subseteq V$ of size $|S| = \sigma$.

\STATE Let $\pi$ be a consistent stable $1$-restorable RPTS.

\FORALL{$s \in S$}
	\STATE Compute an outgoing shortest path tree $T_s$, rooted at $s$, using the (non-faulty) shortest paths selected by $\pi$.
\ENDFOR

\FORALL{$s_1, s_2 \in S$}
	\STATE Using Theorem \ref{thm:singlepairalg}, solve $\srp$ on vertices $\{s_1, s_2\}$ in the graph $T_{s_1} \cup T_{s_2}$.
\ENDFOR

\end{algorithmic}
\end{algorithm}

\begin{theorem}
Given an $n$-vertex, $m$-edge graph $G$ and $|S| = \sigma$ source vertices, Algorithm \ref{alg:srp} solves $\srp$ in $O(\sigma m)$ + $\Oish(\sigma^2 n)$ time.
\end{theorem}
\begin{proof}
First we show correctness.
Since $\pi$ is $1$-restorable, for any $s_1, s_2 \in S$ and failing edge $e$, there exists a vertex $x$ such that $\pi(s_1, x) \cup \pi(s_2, x)$ form a shortest $s_1 \leadsto s_2$ path in $G \setminus \{e\}$.
Since $\pi(s_1, x), \pi(s_2, x)$ are both contained in the graph $T_{s_1} \cup T_{s_2}$, the algorithm correctly outputs a shortest $s_1 \leadsto s_2$ replacement path avoiding $e$.

We now analyze runtime.
Computing $\pi$ takes $O(m)$ time, under Theorem \ref{thm:rwtexists}, due to the selection of random edge weights.\footnote{We analyze in the real-RAM computational model, but by the discussion in Section \ref{sec:detbit}, the (randomized) time to compute $\pi$ would be $\Oish(m)$ if one attends to bit precision.}
Using Dijkstra's algorithm, it takes $O(m + n \log n)$ time to compute an outgoing shortest path tree $T_s$ in $G'$ for each of the $|S| = \sigma$ source vertices, which costs $O(\sigma(m + n \log n))$ time in total.
Finally, we solve $\srp$ for each pair of vertices $s_1, s_2 \in S$, on the graph $T_{s_1} \cup T_{s_2}$ which has only $O(n)$ edges.
By Theorem \ref{thm:singlepairalg} each pair requires $\Oish(n)$ time, for a total of $\Oish(\sigma^2 n)$ time.
\end{proof}

\subsection{Distance Labeling Schemes}
In this section, we show how to use our $f$-restorable RPTSes to provide fault-tolerant (exact) distance labels of sub-quadratic size. We have:
\begin{theorem} 
For any fixed nonnegative integer $f \geq 0$, and $n$-vertex unweighted undirected graph, there is an $(f+1)$-FT distance labeling scheme that assigns each vertex a label of
$$O\left(n^{2 - 1/2^f} \log n\right) \text{ bits}.$$
\end{theorem}
\begin{proof}
Let $\pi$ be a consistent stable restorable $f$-RPTS from Theorem \ref{thm:restorabletiebreaking} on the input graph $G = (V, E)$.
For each vertex $s \in V$, compute an $f$-FT $\{s\} \times V$ preserver by overlaying the $\{s\} \times V$ replacement paths selected by $\pi$ with respect to $\le f$ edge failures.
The label of $s$ just explicitly stores the edges of this preserver; the bound on label size comes from the number of edges in Theorem \ref{thm:expftbfs} and the fact that $O(\log n)$ bits are required to describe each edge.

To compute $\dist_{G \setminus F}(s, t)$, we simply read the labels of $s, t$ to determine their associated preservers and we union these together.
By the $f$-restorability of $\pi$, there exists a valid $s \leadsto t$ replacement path avoiding $F$ formed by concatenating a path included in the preserver of $s$ and a path included in the preserver of $t$.
So, it suffices to remove the edges of $F$ from the union of these two preservers and then compute $\dist(s, t)$ in the remaining graph.
\end{proof}

\subsection{Applications in Fault-Tolerant Network Design}
We next show applications in fault-tolerant network design.
The following theorem extends Theorem 2 in \cite{bodwin2017preserving} which was shown (by a very different proof) for $f=1$, to any $f\geq 1$.
\begin{theorem} \label{thm:oneftss}
Given an $n$-vertex graph $G = (V, E)$, a set of source vertices $S \subseteq V$, and a fixed nonnegative integer $f$, there is an $(f+1)$-FT $S \times S$ distance preserver of $G, S$ on $O\left( n^{2-1/2^f} |S|^{1/2^f} \right)$ edges.
\end{theorem}
\begin{proof}
Using Theorem \ref{thm:restorabletiebreaking}, let $\pi$ be an $(f+1)$-RPTS that is simultaneously stable, consistent, and $(f+1)$-restorable.
The construction is to build an $f$-FT $S \times V$ preserver by overlaying all $S \times V$ replacement paths selected by $\pi$ with respect to $\le f$ edge faults.
By Theorem \ref{thm:expftbfs}, this has the claimed number of edges.

To prove correctness of the construction, we invoke restorability.
Consider some vertices $s, t \in S$ and set of $|F| \le f+1$ edge faults.
Since $\pi$ is $(f+1)$-restorable, there is a valid $s \leadsto t$ replacement path avoiding $F$ that is the concatenation of two shortest paths of the form $\pi(s, x \mid F'), \pi(t, x \mid F')$, where $|F'| \le f$ and $x \in V$.
These replacement paths are added as part of the $f$-FT preservers, and hence the union of the preservers includes a valid $s \leadsto t$ replacement path avoiding $F$.
\end{proof}

We can then plug these subset distance preservers into a standard application to additive spanners.
We will black-box the relationship between subset preservers and additive spanners, so that it may be applied again when we give distributed constructions below.
The following lemma is standard in the literature on spanners.

\begin{lemma} \label{lem:prestospan}
Suppose that we can construct an $f$-FT $S \times S$ distance preserver on $g(n, \sigma, f)$ edges, for any set of $|S|=\sigma$ vertices in an $n$-vertex graph $G$.
Then $G$ has an $f$-FT $+4$ additive spanner on $O(g(n, \sigma, f) + nf + n^2 f/\sigma)$ edges.
\end{lemma}
\begin{proof}
For simplicity, we will give a randomized construction where the error bound holds deterministically and the edge bound holds in expectation.
Naturally, one can repeat the construction $O(\log n)$ times and select the sparsest output spanner to boost the edge bound to a high-probability guarantee.

Let $C \subseteq V$ be a subset of $\sigma$ vertices, selected uniformly at random.
Call the vertices in $C$ \emph{cluster centers}.
For each vertex $v \in V$, if it has at least $f+1$ neighbors in $C$, then add an arbitrary $f+1$ edges connecting $v$ to vertices in $C$ to the spanner $H$, and we will say that $v$ is \emph{clustered}.
Otherwise, if $v$ has $\le f$ neighbors in $C$, then add all edges incident to $v$ to the spanner, and we will say that $v$ is \emph{unclustered}.
The second and final step in the construction is to add an $f$-FT subset distance preserver over the vertex set $C$ to the spanner $H$.

\paragraph{Spanner Size.}
By hypothesis, the subset distance preserver in the construction costs $g(n, \sigma, f)$ edges, which gives the first term in our claimed edge bound.
Consider an arbitrary node $v$, and let us count the expected number of edges added due to $v$ in the first step of the construction.
We add $O(f)$ edges per \emph{clustered} nodes, for $O(nf)$ edges in total, which gives the second term in our edge bound.

We then analyze the edges added due to unclustered nodes; this part is a little more complicated, but still standard in the area (e.g., \cite{BKMP10}).
Let $v$ be an arbitrary node, let $N(v)$ be its neighborhood, and suppose $\deg(v) \ge 2f$.
Our goal is to bound the expected number of edges contributed by $v$ being unclustered, which is the quantity
$$\Pr\left[ |\{c \in N(v) \ \mid \ c \text{ cluster center}\}| \le f\right] \cdot \deg(v).$$
To compute this probability, we can imagine $\sigma$ sequential experiments, in which the next cluster center is randomly selected among the previously un-selected nodes.
The experiment is ``successful'' if the cluster center is one of the $\Theta(\deg(v))$ unselected nodes in $N(v)$, or ``unsuccessful'' if it is one of the $\Theta(n)$ other unselected nodes.
Thus each experiment succeeds with probability $\Theta(\deg(v)/n)$.
By standard Chernoff bounds, the probability that $\le f$ experiments succeed is
$$O\left(\frac{f}{\deg(v)} \cdot \frac{n}{\sigma}\right).$$
Thus the expected number of edges added due to $v$ unclustered is $O(nf/\sigma)$, and by unioning over the $n$ nodes in the graph, the total is $O(n^2 f/\sigma)$.

\paragraph{Spanner Correctness.}

Now we prove that $H$ is an $f$-EFT $+4$ spanner of $G$ (deterministically).
Consider any two vertices $s, t$ and a set of $|F| \le f$ edge faults, and let $q = q(s, t \mid F)$ be any replacement path between them.
Let $x$ be the first clustered vertex and $y$ the last clustered vertex in $q$.
Let $c_x, c_y$ be cluster centers adjacent to $x, y$, respectively, in the graph $G \setminus F$ (since $x, y$ are each adjacent to $f+1$ cluster centers initially, at least one adjacency still holds after $F$ is removed).
We then have:
\begin{align*}
\dist_{H \setminus F}(s, t) &\le \dist_{H \setminus F}(s, x) + \dist_{H \setminus F}(x, y) + \dist_{H \setminus F}(y, t)\\
&= \dist_{G \setminus F}(s, x) + \dist_{H \setminus F}(x, y) + \dist_{G \setminus F}(y, t) \tag*{all unclustered edges in $H$}\\
&\le \dist_{G \setminus F}(s, x) + \left(2 + \dist_{H \setminus F}(c_x, c_y) \right) + \dist_{G \setminus F}(y, t) \tag*{triangle inequality}\\
&= \dist_{G \setminus F}(s, x) + \left(2 + \dist_{G \setminus F}(c_x, c_y) \right) + \dist_{G \setminus F}(y, t) \tag*{$C \times C$ preserver}\\
&\le \dist_{G \setminus F}(s, x) + \left(4 + \dist_{G \setminus F}(x, y) \right) + \dist_{G \setminus F}(y, t) \tag*{triangle inequality}\\
&= \dist_{G \setminus F}(s, t) + 4.
\end{align*}
where the last equality follows since $x, y$ lie on a valid $s \leadsto t$ replacement path.
\end{proof}

Using this, we get:

\begin{theorem}\label{thm:additive}
For any $n$-vertex graph $G = (V, E)$ and nonnegative integer $f$, there is an $(f+1)$-FT $+4$ additive spanner on $O_f\left(n^{1 + 2^f / (2^f + 1)}\right)$ edges.
\end{theorem}
\begin{proof}
The construction is to simply applying Lemma \ref{lem:prestospan} to the subset distance preservers from Theorem \ref{thm:oneftss}, balancing parameters by choosing $\sigma := n^{1/(2^f + 1)}$.
From Theorem \ref{thm:oneftss}, the size of the subset distance preserver is
\begin{align*}
O_f\left(n^{2 - 1/2^f} \cdot \left(n^{1/(2^f + 1)}\right)^{1/2^f}\right) &= O_f\left(n^{2 - 1/2^f + 1/(2^f \cdot (2^f + 1))}\right)\\
&= O_f\left(n^{2 - 1/(2^f + 1)}\right)\\
&= O_f\left(n^{1 + 2^f / (2^f + 1)}\right).
\end{align*}
The $O(nf)$ term in Lemma \ref{lem:prestospan} can be ignored, and the $O(n^2 f/\sigma)$ term is again
\begin{align*}
O_f\left(n^{2 - 1/(2^f+1)}\right) = O_f\left(n^{1 + 2^f / (2^f + 1)}\right),
\end{align*}
and the theorem follows.
\end{proof}

\input{distributed-cons.tex}

%% file: distributed-cons.tex
\subsection{Distributed Constructions}\label{sec:dist}
Throughout this section, we consider the standard $\mathsf{CONGEST}$  model of distributed computing \cite{Peleg:2000}. In this model, the network is abstracted as an $n$-vertex graph $G=(V, E)$, with one processor on each vertex. Initially, these processors only know their incident edges in the graph, and the algorithm proceeds in synchronous communication rounds over the graph $G=(V,E)$. In each round, vertices are allowed to exchange $O(\log n)$ bits with their neighbors and perform local computation. Throughout, the diameter of the graph $G=(V,E)$ is denoted by $D$. 

\begin{lemma} [Distributed Tie-Breaking SPT] \label{lem:SPT-tieBreaking}
For every unweighted and undirected $n$-vertex graph $G=(V,E)$, tiebreaking weight function $\omega: E \to [1-\epsilon,1+\epsilon]$, and every source vertex $s$, there is a deterministic algorithm that computes a shortest-path tree rooted at $s$ (based on the weights of $\omega$) within $O(D)$ rounds. The total number of messages sent through each edge is bounded $O(1)$.  
\end{lemma}
\begin{proof}
Let $\dist(u,v)$ denote the unweighted $u$-$v$ distance in $G$, and $\dist^*(u,v)$ denote the (directed) $u \leadsto v$ distance under the weights of $\omega$. Since $\omega$ is only a tie-breaking weight function, any shortest path tree of $s$ under $\omega$ is also a legit BFS tree for $s$. In other words, all vertices at unweighted distance $d$ from the source $s$ must appear on level $d$ in the SPT of $s$ under $\omega$. The construction therefore is almost analogous to the standard distributed BFS construction, and the only difference is that each vertex uses the weights of $\omega$ in order to pick its parent in the tree.

The algorithm works in $O(D)$ steps, each step is implemented within $O(1)$ rounds. The invariant at the beginning of phase $i \in \{1,\ldots, D\}$ is as follows: all vertices in the first $i$ layers $L_0,\ldots, L_{i-1}$ of the SPT are marked, and each of these vertices $v$ know their distance $\dist^*(s,v)$. This clearly holds for $i=1$ as $L_0=\{s\}$ and $\dist^*(s,s)=0$. In phase $i$, all vertices $u$ of layer $L_i$ broadcast the distance $\dist^*(s,u)$ to their neighbors. Every vertex $v \notin \bigcup_{j=0}^{i-1} L_j$ that receives messages from its neighbors in layer $L_i$ picks its parent $u$ to be the vertex that minimizes its $\dist^*(s,v)$ distance. That is, $u=\arg\min_{w \in L_i \cap N(v)}\dist^*(s,w)+\omega(w,v)$. It is easy to see by induction on $i$, that the invariant is now satisfied at the beginning of phase $i+1$. After $D$ rounds, the construction of the tree is completed. As in the standard BFS construction, only $O(1)$ number of messages are sent through each edge in the graph.
\end{proof}

We make an extensive use of the random delay approach of \cite{leighton1994packet,Ghaffari15}. Specifically, we use the following theorem:
\begin{theorem}[{\cite[Theorem 1.3]{Ghaffari15}}]\label{thm:delay}
Let $G$ be a graph and let $A_1,\ldots,A_m$ be $m$ distributed algorithms in 
the $\mathsf{CONGEST}$ model,  where each algorithm takes at most $\dilation$ rounds, and where for each 
edge of $G$, at most $\congestion$ messages need to go through it, in total 
over all these algorithms. Then, there is a randomized distributed
algorithm (using only private randomness) that, with high probability, 
produces 
a schedule that runs all the algorithms in $O(\congestion +\dilation \cdot 
\log 
n)$ rounds, after $O(\dilation \log^2 n)$ rounds of pre-computation.
\end{theorem}

Using the random delay approach for computing SPT with respect to our $1$-restorable tie-breaking scheme, provides an efficient distributed construction of $1$-FT $S \times S$ preserves, and consequently also $1$-FT $+4$-additive spanners.  
\begin{lemma}[Dist. $1$-FT $S \times S$ Preserver]\label{lemma:onefault-SxS}
For every unweighted and undirected $n$-vertex graph $G=(V,E)$ and subset of sources $S$, there is a randomized algorithm that computes a $1$-FT $S \times S$ preserver with $O(|S|n)$ edges within $\widetilde{O}(D+|S|)$ rounds, with high probability.
\end{lemma}
\begin{proof}
First, the vertices locally compute the $1$-restorable tie-breaking weight function $\omega$, by letting each vertex $u$ sample the weights for its incident edges, and sending it to the second edge endpoints. This is done in a single communication round. Then we apply the SPT construction of Lemma \ref{lem:SPT-tieBreaking} under $\omega$, for every source $s \in S$. We run all these algorithms, $A_{s_1},\ldots, A_{s_\sigma}$, simultaneously in parallel using the random delay approach. By sharing a shared seed of $O(\log^2 n)$ bits, each vertex can compute the starting time of each algorithm $A_{s_i}$. As the total congestion of these algorithms is $O(|S|)$, by Thm. \ref{thm:delay}, the round complexity is bounded by $\widetilde{O}(D+|S|)$ rounds, w.h.p.
\end{proof}

By applying the constructions of $1$-FT $S \times V$ preservers and $2$-FT $S \times V$ preservers of \cite{ParterDualDist20} and using the $1$-restorable weight function to break the BFS ties, we immediately get 
$2$-FT $S \times S$ preservers and $3$-FT $S \times S$ preservers. Specifically, Theorem \ref{thm:dist-constructions} (2) (resp., (3)) follows by using the $1$-restorable weight function with Theorem 1 (resp., Theorem 2) of \cite{ParterDualDist20}. E.g., to compute the $2$-FT $S \times S$ preserver, we apply the construction of $1$-FT $S \times V$ preservers of Theorem 1 in \cite{ParterDualDist20} with the only difference being that the shortest path ties are decided based on the $1$-restorable weight function instead of breaking it arbitrarily. This can be easily done by augmenting the BFS tokens with the length of the path from the root.

Finally, the construction of FT $S \times S$ preservers can naturally yield FT $+4$-additive spanners. No prior constructions of such spanners have been known before (not even for $f=1$). 
%

\begin{proof}[Proof of Cor. \ref{cor:dist-add-constructions}]
%
Let $S$ be a sample of $\sigma=\Theta(\sqrt{n}\log n)$ sources sampled independently in $V$. 
Item (1) of the corollary follows by applying the construction of Theorem \ref{thm:dist-constructions}(1) and the correctness follows by Lemma \ref{lem:prestospan}.  In the same manner, item (2) of the corollary follows by applying the above construction with $\sigma=n^{1/3}$ sources, and using Theorem \ref{thm:dist-constructions}(2). Finally, item (3) of the corollary follows by applying Theorem \ref{thm:dist-constructions} (3) with $\sigma=n^{1/9}$. 
\end{proof}

%% file: lower-bound.tex
\section{Impossibility of Symmetry and Restorability}

As observed by Afek et al.~\cite{ABKCM02}, one cannot generally have restorability and symmetry at the same time:

\begin{theorem}\label{thm:impossible}
There are input graphs that do not admit a tiebreaking scheme that is simultaneously symmetric and $1$-restorable.
\end{theorem}
\begin{proof}
The simplest example is a $C_4$:
\begin{center}
\begin{tikzpicture}

\draw [fill=black] (3, 0) circle [radius=0.15];
\draw [fill=black] (4, 0) circle [radius=0.15];
\node at (3.5, 0) {\Huge $\mathbf \times$};
\node at (3, -0.5) {$s$};
\node at (4, -0.5) {$t$};

\draw [fill=black] (3, 1) circle [radius=0.15];
\draw [fill=black] (4, 1) circle [radius=0.15];
\node at (3, 1.5) {$x$};
\node at (4, 1.5) {$y$};


\draw (3, 0) -- (3, 1) -- (4, 1) -- (4, 0) -- cycle;

\end{tikzpicture}
\end{center}
Assume $\pi$ is symmetric and consider the selected non-faulty shortest paths $\pi(s, y)$ and $\pi(x, t)$ going between the two opposite corners.
These paths must intersect on an edge; without loss of generality, suppose this edge is $(s, t)$.
Then $\pi(s, t)$ is just the single edge $(s, t)$.
Suppose this edge fails, and so the unique replacement $s \leadsto t$ path is $q = (s, x, y, t)$.
Since both $\pi(s, y)$ and $\pi(x, t)$ use the edge $(s, t)$, this path does not decompose into two non-faulty shortest paths selected by $\pi$.
Hence $\pi$ is not $1$-restorable.
\end{proof}

\section{Lower bound for $f$-failures preservers with a consistent and stable tie-breaking scheme \label{app:conslb}}
In this section we prove Theorem \ref{thm:cslowerbound} by giving a lower bound constructions for $S \times V$ distance preservers using a consistent and stable tie-breaking scheme. 


We begin by showing the construction for the single source case (i.e., $\sigma=1$) and then extend it to the case of multiple sources.
Our construction is based on the graph $G_f(d)=(V_f,E_f)$, defined inductively. 
For $f=1$, $G_1(d)=(V_1, E_1)$ consists of three components:
\begin{enumerate}[noitemsep]
\item a set of vertices $U=\{u^1_1,\ldots,u^1_d\}$ connected by a path
$P_1=[u^1_1, \ldots, u^1_d]$,
\item a set of terminal vertices $Z=\{z_1,\ldots,z_d\}$
(viewed by convention as ordered from left to right),
\item a collection of $d$ vertex disjoint paths $\{Q^1_{i}\}$, 
where each path $Q^1_{i}$ connects $u^1_i$ and $z_i$ and has length of $d-i+1$ edges, for every $i \in \{1, \ldots, d\}$. 
\end{enumerate}
The vertex $\Root(G_1(d))=u^1_1$ is fixed as the root of $G_1(d)$, hence
the edges of the paths $Q^1_i$ are viewed as directed away from $u^1_i$,
and the terminal vertices of $Z$ are viewed as the \emph{leaves} of the graph,
denoted $\Leaf(G_1(d))=Z$.
Overall, the vertex and edge sets of $G_1(d)$ are
$V_1=U \cup Z \cup \bigcup_{i=1}^d V(Q^1_i)$ and
$E_1=E(P_1) \cup \bigcup_{i=1}^d E(Q^1_i)$.

For ease of future analysis, we assign labels to the leaves
$z_i \in \Leaf(G_1(d))$.
Let $\LAB_f: \Leaf(G_f(d))  \to E(G_f(d))^f$.
The label of each leaf corresponds to a set of edge faults under which
the path from root to leaf is still maintained (as will be proved later on).
Specifically, $\LAB_1(z_i, G_1(d))=(u^1_i,u^1_{i+1})$ for $i \in [1,d-1]$.
In addition, define
$P(z_i,G_1(d)) = P_1[\Root(G_1(d)),u^1_i] \circ Q^1_i$
to be the path from the root $u^1_1$ to the leaf $z_i$.

To complete the inductive construction, let us describe the construction
of the graph $G_{f}(d)=(V_{f}, E_{f})$, for $f\ge 2$,
given the graph $G_{f-1}(\sqrt{d})=(V_{f-1}, E_{f-1})$.
The graph $G_{f}(d)=(V_{f}, E_{f})$ consists of the following components.
First, it contains a path $P_f=[u^f_1, \ldots, u^f_d]$, where
the vertex $\Root(G_{f}(d))=u^f_1$ is fixed to be the root.
In addition, it contains $d$ disjoint copies of the graph $G'=G_{f-1}(\sqrt{d})$,
denoted by $G'_1, \ldots, G'_d$
(viewed by convention as ordered from left to right),
where each $G'_i$ is connected to $u^f_i$ by a collection of $d$
vertex disjoint paths $Q^f_i$, for $i \in \{1, \ldots, d\}$,
connecting the vertices $u^f_i$ with $\Root(G'_i)$.
The length of $Q^f_i$ is  $d-i+1$, and 
the leaf set of the graph $G_{f}(d)$ is the union of the leaf sets of $G'_j$'s,
$\Leaf(G_{f}(d))=\bigcup_{j=1}^d \Leaf(G'_j)$.

Next, define the labels $\LAB_f(z)$ for each $z \in \Leaf(G_{f}(d))$.
For every $j \in \{1, \ldots, d\}$ and any leaf $z_{j,i} \in \Leaf(G'_j)$,
let $\LAB_f(z_{j,i}, G_{f}(d))=(u^f_j,u^f_{j+1}) \circ \LAB_{f-1}(z_{j,i}, G'_j)$.

Denote the size (number of vertices) of $G_f(d)$ by $\NodesIn(f,d)$,
its depth (maximum distance between the root vertex $\Root(G_f(d))$ to a leaf vertex in $\Leaf(G_f(d))$) by  $\depth(f,d)$, and its number of leaves by  $\NLeaf(f,d) = |\Leaf(G_f(d))|$.
Note that for $f=1$,
$\NodesIn(1,d) = 2d+d^2 \leq 2d^2$,
$\depth(1,d)=d$ and $\NLeaf(1,d)=d$.
We now observe that the following inductive relations hold.
\begin{observation}
\label{obs:rel}
(a) $\depth(f,d)=O(d)$, (b) $\NLeaf(f,d)=d^{2-1/2^{f-1}}$ and (c) $\NodesIn(f,d)=2f\cdot  d^2$.
\end{observation}
\begin{proof}
(a) follows by the length of $Q^f_i$, which implies that
$\depth(f,d)=d+\depth(f-1,\sqrt{d})\leq 2d$.
(b) follows by the fact that the terminals of the paths starting with
$u_1^f, \ldots, u_d^f$ are the terminals of the graphs $G'_1, \ldots, G'_d$
which are disjoint copies of $G_{f-1}(\sqrt{d})$, so $\NLeaf(f,d)=d \cdot \NLeaf(f-1,\sqrt{d})$.
(c) follows by summing the vertices in the $d$ copies of $G'_i$
(yielding $d \cdot \NodesIn(f,d)$) and the vertices in $d$ vertex disjoint paths,
namely $Q^f_1, \ldots, Q^f_d$ of total $d^2$ vertices,
yielding $\NodesIn(f,d)=d \cdot \NodesIn(f-1,\sqrt{d})+d^2\leq 2fd^2$.
\end{proof}
Consider the set of leaves in $G_f(d)$, namely,
$\Leaf(G_f(d)) = \bigcup_{i=1}^d \Leaf(G'_i) = \{z_1, \ldots, z_\lambda\}$,
ordered from left to right according to their appearance in $G_f(d)$.

For every leaf vertex $z \in \Leaf(G_f(d))$, we define inductively a path $P(z, G_f(d))$ connecting the root $\Root(G_{f}(d))=u^f_1$ with the leaf $z$. As described above for $f=1$, $P(z_i,G_1(d)) = P_1[\Root(G_1(d)),u^1_i] \circ Q^1_i$. Consider a leaf $z \in \Leaf(G_f(d))$ such that $z$ is the $i^{th}$ leaf in the graph $G'_j$. We therefore denote $z$ as $z_{i,j}$, and define $P(z_{j,i},G_f(d)) = P_f[\Root(G_f(d)),u^1_j] \circ Q^f_j \circ P(z_{j,i},G'_j)$. We next claim the following on these paths.

\begin{lemma}
\label{lem:prop_induc_path}
For every leaf $z_{j,i} \in \Leaf(G_f(d))$ it holds that: \\
(1) The path  $P(z_{j,i}, G_f(d))$ is the only $u^f_1-z_{j,i}$ path in $G_f(d)$.\\
(2) $P(z_{j,i}, G_f(d)) \subseteq G \setminus \bigcup_{i \geq j}\LAB_f(z_{j,i}, G_f(d)) \cup \bigcup_{k \geq j, \ell \in [1,\NLeaf(f-1,\sqrt{d})]}\LAB_f(z_{k,\ell}, G_f(d))$.\\
(3) $P(z_{j,i}, G_f(d)) \not\subseteq G \setminus \LAB_f(z_{k,\ell}, G_f(d))$
for $k<j$ and every $\ell \in [1,\NLeaf(f-1,\sqrt{d})]$, as well as for $k= j$ and every $\ell\in [1, i-1]$.
(4) $|P(z, G_f(d))| = |P(z', G_f(d))|$ for every $z,z' \in \Leaf(G_f(d))$. 
\end{lemma}
\begin{proof}
We prove the claims by induction on $f$.
For $f=1$, the lemma holds by construction.
Assume this holds for every $f' \leq f-1$ and consider $G_f(d)$.
Recall that $P_f=[u^f_1, \ldots, u^f_d]$, and let $G'_1, \ldots, G'_d$ be $d$ copies
of the graph $G_{f-1}(\sqrt{d})$, viewed as ordered from left to right,
where $G'_j$ is connected to $u^f_j$. That is, there are disjoint paths $Q^f_j$ 
connecting $u^f_j$ and $\Root(G'_j)$, for every $j\in \{1,\ldots, d\}$.

Consider a leaf vertex $z_{j,i}$, the $i^{th}$ leaf vertex in $G'_j$. By the inductive assumption, there exists a single path $P(z_{j,i}, G'_j)$
between the root $\Root(G'_j)$ and the leaf $z_{j,i}$, for every $j \in \{1,\ldots, d\}$.
We now show that there is a single path between $\Root(G_f(d)) = u^f_1$
and $z_{j,i}$ for every $j\in \{1,\ldots, d\}$.
Since there is a single path $P'$ connecting $\Root(G_f(d))$ and $\Root(G'_j)$ given by $P'=P_f[u^f_1, u^f_j]\circ Q^f_j$, it follows that
$P(z_{j,i}, G_f(d))=P' \circ P(z_{j,i}, G'_j)$ is a unique path in $G_f(d)$.

We now show (2). We first show that $P(z_{j,i}, G_f(d)) \subseteq G \setminus \bigcup_{\ell \geq i ~\mid~ z_{j,\ell} \in \Leaf(G'_j)} LAB_f(z_{j,\ell}, G_f(d))$.  By the inductive assumption,
$P(z_{j,i}, G'_j) \in G \setminus \bigcup_{\ell \geq i} \LAB_{f-1}(z_{j,\ell}, G'_j)$.
Since $\LAB_f(z_{j,i}, G_f(d))=(u^f_{j}, u^f_{j+1}) \circ \LAB_{f-1}(z_{j,i}, G'_j)$,
it remains to show that $e_\ell=(u^f_{\ell}, u^f_{\ell+1})  \notin P'$ for $\ell \geq i$.
Since $P'$ diverges from $P_f$ at the vertex $u^f_j$,
it holds that $e_j, \ldots, e_{d-1} \notin P(z_{j,i}, G_f(d))$. We next complete the proof for every leaf vertex $z_{k,\ell}$ for $z_{k,\ell} \in \Leaf(G'_q)$ for $k > j$ and every $\ell \in \NLeaf(f-1,\sqrt{d})$. The claim holds as the edges of $G'_j$ and $G'_k$ are edge-disjoint, and $e_j, \ldots, e_{d-1} \notin P(z_{j,i}, G_f(d))$. 

Consider claim (3) and a leaf vertex $z_{j,i} \in \Leaf(G'_j)$ for some $j \in \{1,\ldots, d\}$ and $i \in \NLeaf(f-1,\sqrt{d})$. Let $Z_1=\{z_{j,\ell} \in  \Leaf(G'_j) \mid \ell < i\}$ be the set
of leaves to the left of $z_{j,i}$ that belong to $G'_j$, and let
$Z_2=\{z_{k,\ell} \notin  \Leaf(G'_j) \mid j > k\}$ be
the complementary set of leaves to the left of $z_{j,i}$. By the inductive assumption,
$P(z_{j,i}, G'_j) \nsubseteq G \setminus \LAB_{f-1}(z_{j,\ell}, G'_j)$ for every $z_{j,\ell} \in Z_1$.
The claim holds for $Z_1$ as the order of the leaves in $G'_j$ agrees with their order in $G_f(d)$, and $\LAB_{f-1}(z_{k,\ell}, G'_j) \subset \LAB_{f}(z_{k,\ell}, G_f(d))$. 

Next, consider the complementary leaf set $Z_2$ to the left of $z_{j,i}$.
Since for every $z_{k,\ell} \in Z_2$, the divergence point of $P(z_{k,\ell}, G_f(d))$ and $P_f$ is at $u^f_k$
for $k < j$, it follows that $e_k=(u^f_k, u^f_{k+1}) \in P(z_{j,i}, G_f(d))$,
and thus $P(z_{j,i}, G_f(d)) \nsubseteq G \setminus \LAB_f(z_{k,\ell}, G_f(d))$
for every $z_{k,\ell} \in Z_2$. Finally, consider (4). By setting the length of the paths $Q^f_j$ to $d-j+1$ for every $j \in \{1,\ldots, d\}$, we have that $\dist(u^f_1, \Root(G'_j))=d$ for every $j \in [1,d]$. The proof then follows by induction as well, since $|P(z_{j,i}, G'_j)|=|P(z_{k,\ell}, G'_k)|$ for every $k, j \in [1,d]$ and $i,\ell \in [1, \NLeaf(f-1,\sqrt{d}]$.
\end{proof}
Finally, we turn to describe the graph $G^*_f(V, E,W)$ which establishes our
lower bound, where $W$ is a particular bad edge weight function that determines the consistent tie-breaking scheme which provides the lower bound. The graph $G^*_f(V, E,W)$ consists of three components.
The first is the graph $G_{f}(d)$ for $d=\lfloor \sqrt{n/(4f)} \rfloor$.
By Obs. \ref{obs:rel}, $\NodesIn(f,d)=|V(G_{f}(d))|\leq n/2$.
The second component of $G^*_f(V, E,W)$ is a set of vertices
$X=\{x_1, \ldots, x_\chi\}$, where the last vertex of $P_f$, namely, $u^f_d$ is
connected to all the vertices of $X$.
The cardinality of $X$ is $\chi=n-\NodesIn(f,d)-1$.
The third component of $G^*_f(V, E,W)$ is a complete bipartite graph $B$
connecting the vertices of $X$ with the leaf set $\Leaf(G_f(d))$, i.e.,
the disjoint leaf sets $\Leaf(G'_1), \ldots, \Leaf(G'_d)$.
We finally define the weight function $W: E \to (1,1+1/n^2)$. Let $W(e)=1$ for every $e \in E \setminus E(B)$. 
The weights of the bipartite graph edges $B$ are defined as follows. Consider all leaf vertices $\Leaf(G_f(d))$ from left to right given by $\{z_1, \ldots, z_\lambda\}$. Then, $W(z_j,x_i)=(\lambda-j)/n^4$ for every $z_j$ and every $x_i \in X$. 
The vertex set of the resulting graph is thus
$V=V(G_{f}(d))\cup \{v^{*}\} \cup X$ and hence $|V|=n$.
By Prop. (b) of Obs. \ref{obs:rel},
$\NLeaf(G_f(d))=d^{2-1/2^{f-1}}=\Theta((n/f)^{1-1/2^f}),$
hence $|E(B)|=\Theta((n/f)^{2-1/2^f})$.

We now complete the proof of Thm. \ref{thm:cslowerbound} for the single source case.
\begin{proof}[Thm. \ref{thm:cslowerbound} for $|S|=1$.]
Let $s=u^f_1$ be the chosen source in the graph $G^*_f(V, E,W)$. 
We first claim that under the weights $W$, there is a unique shortest path, denoted by $\pi(s,x_i ~\mid~ F)$ for every $x_i \in X$ and every fault set $F \in \{\LAB_f(z_1, G_f(d)), \ldots, \LAB_f(z_\ell, G_f(d))\}$. By Lemma \ref{lem:prop_induc_path}(1), there is a unique shortest path from each $s$ to each $z_j \in \Leaf(G_f(d))$ denoted by $P(z_j, G_f(d))$. 

In addition, by Lemma \ref{lem:prop_induc_path}(4), the unweighted length of all the $s$-$z_j$ paths are the same for every $z_j$. Since each $x_i$ is connected to each $z_j$ with a distinct edge weight in $(1,1+1/n^2)$, we get that each $x_i$ has a unique shortest path from $s$ in each subgraph $G \setminus \LAB_f(z_j, G_f(d))$. Note that since the uniqueness of $\pi$ is provided by the edge weights it is both consistent and stable. Also note that the weights of $W$ are sufficiently small so that they only use to break the ties between equally length paths.

We now claim that a collection of $\{s\} \times X$ replacement paths (chosen based on the weights of $W$) contains all edges of the bipartite graph $B$. Formally, letting 
$$\mathcal{P}=\bigcup_{x_i \in X} \bigcup_{z_j \in \Leaf(G_f(d))} \pi(s,x_i ~\mid~ \LAB_f(z_j, G_f(d)))~,$$
we will show that $E(B) \subseteq \bigcup_{P \in \mathcal{P}} P$ which will complete the proof. 
To see this we show that $\pi(s,x_i ~\mid~ \LAB_f(z_j, G_f(d)))=P(z_j, G_f(d)) \circ (z_j,x_i)$. Indeed, by Lemma \ref{lem:prop_induc_path}(2), we have that $P(z_j, G_f(d))  \subseteq G \setminus \LAB_f(z_j, G_f(d))$. It remains to show that the shortest $s$-$x_i$ path (based on edge weights) in $G \setminus \LAB_f(z_j, G_f(d))$ goes through $z_j$. 
By Lemma \ref{lem:prop_induc_path}(2,3), the only $z_k$ vertices in $\Leaf(G_f(d))$ that are connected to $s$ in $G \setminus \LAB_f(z_j, G_f(d))$ are $\{z_1,\ldots, z_j\}$. Since $W(z_1,x_i) > W(z_2,x_i)> \ldots > W(z_j,x_i)$, we have that $(z_j, x_i)$ is the last edge of $\pi(s,x_i ~\mid~ \LAB_f(z_j, G_f(d)))$. As this holds for every $x_i \in X$ and every $z_j \in \Leaf(G_f(d))$, the claim follows. 
\end{proof}


\paragraph{Extension to multiple sources.}
Given a parameter $\NSource$ representing the number of sources,
the lower bound graph $G$ includes $\NSource$ copies, $G'_1, \ldots, G'_\NSource$, of $G_f(d)$,
where $d=O(\sqrt{(n / 4f\NSource)})$.
By Obs. \ref{obs:rel}, each copy consists of at most $n/2\NSource$ vertices.
We now add to $G$ a collection $X$ of $\Theta(n)$ vertices connected to the $\NSource$ leaf sets
$\Leaf(G'_1), \ldots, \Leaf(G'_\NSource)$ by a complete bipartite graph $B'$. See Fig. \ref{fig:LB-graph-induc} for an illustration.
 We adjust the size of the set $X$ in the construction so that $|V(G)|=n$.
Since $\NLeaf(G'_i)=\Omega((n / (f\NSource))^{1-1/2^f})$ (see Obs. \ref{obs:rel}),
overall $|E(G)| = \Omega(n \cdot \NSource \cdot \NLeaf(G_f(d))) =
\Omega(\NSource^{1/2^f}\cdot (n/f)^{2-1/2^f})$.
The weights of all graph edges not in $B'$ are set to $1$. For every $i \in \{1,\ldots, \NSource\}$, the edge weights of the bipartite graph $B_j=(\Leaf(G'_1),X)$ are set in the same manner as for the single source case.
Since the path from each source $s_i$ to $X$ cannot aid the vertices of $G'_j$
for $j \neq i$, the analysis of the single-source case can be applied
to show that each of the bipartite graph edges in necessary
upon a certain sequence of at most $f$-edge faults.
This completes the proof of Thm. \ref{thm:cslowerbound}.

\begin{figure}[h!]
\begin{center}
\includegraphics[scale=0.40]{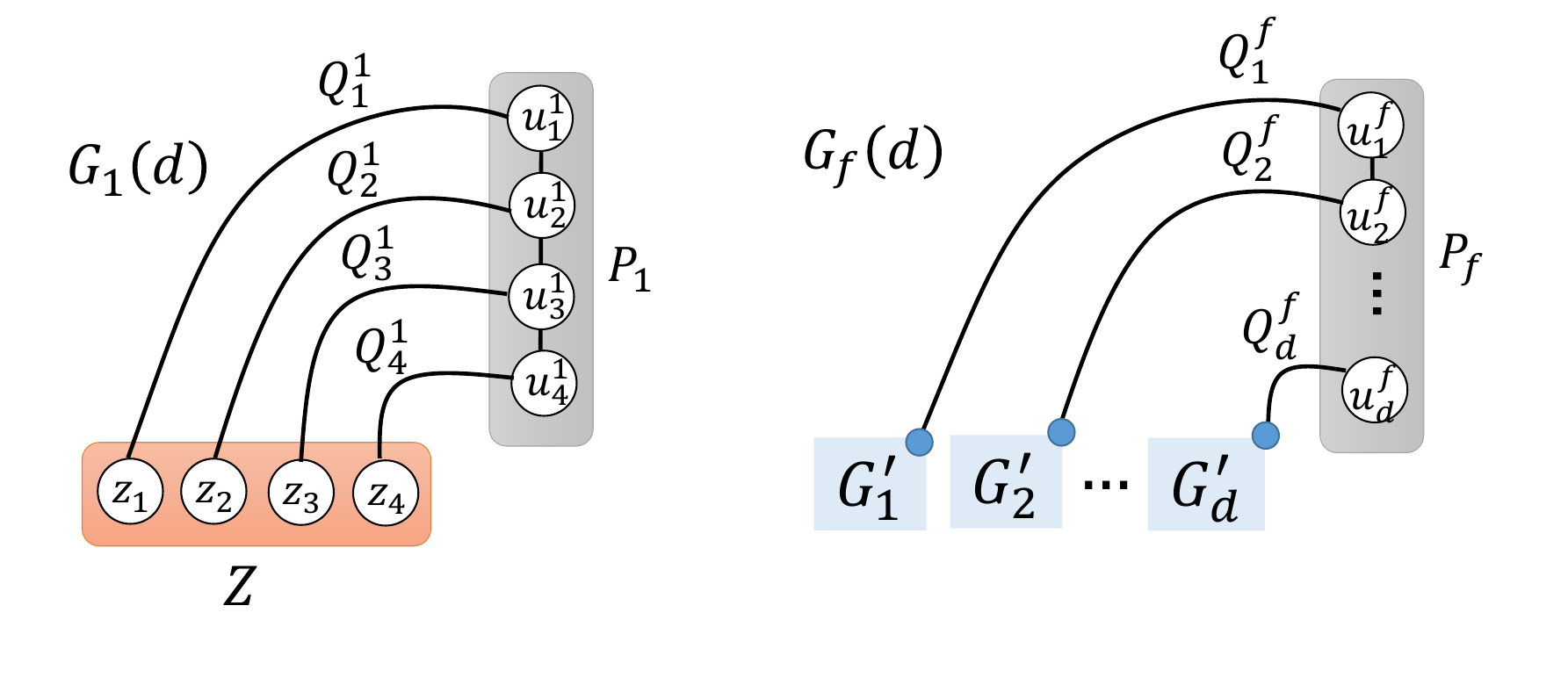}
\includegraphics[scale=0.40]{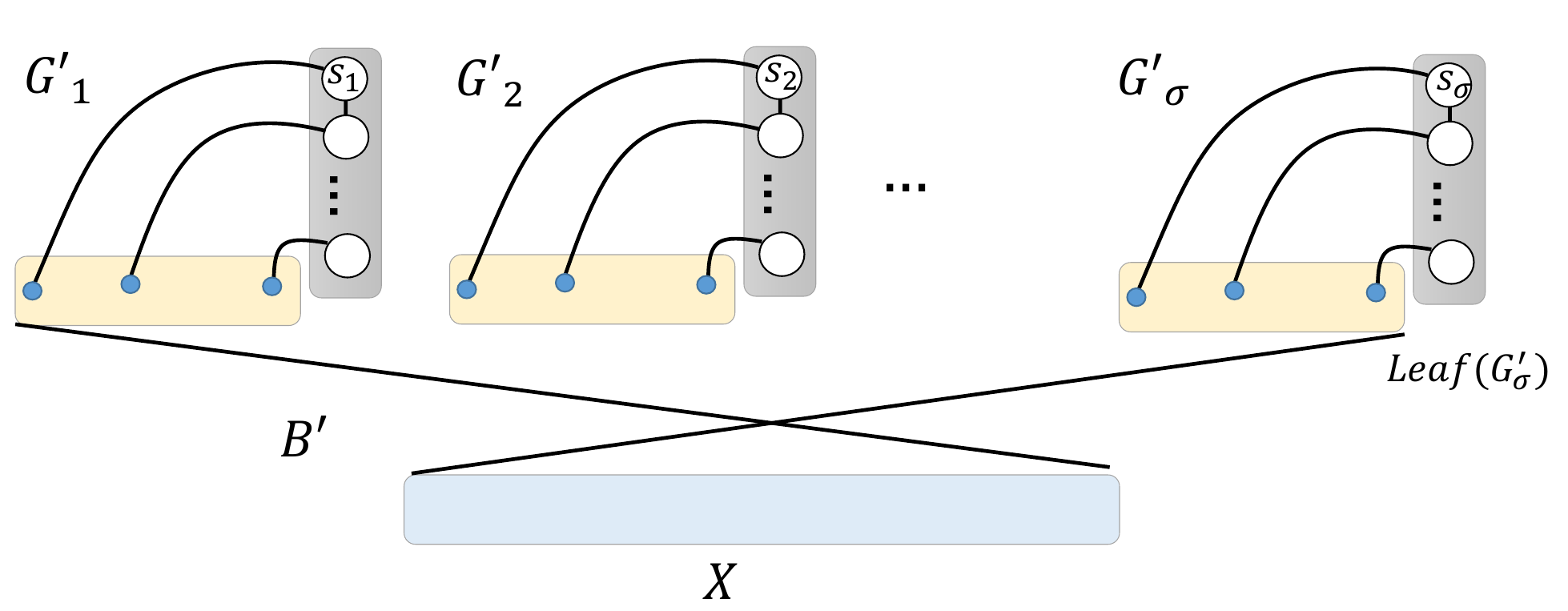}
\caption{\sf Top: Illustration of the graphs $G_1(d)$ and $G_f(d)$. Each graph $G'_i$ is a graph of the form $G_f(\sqrt{d})$. Bottom: Extension to $\sigma$ sources. The collection of leaf nodes $\Leaf(G'_1), \ldots, \Leaf(G'_\NSource)$ are fully connected to a linear size set $X$. The size of the resulting bipartite graph $B'$ dominates the size of the construction. \label{fig:LB-graph-induc}
}
\end{center}
\end{figure}

\begin{figure}[h!]
\begin{center}
\includegraphics[scale=0.40]{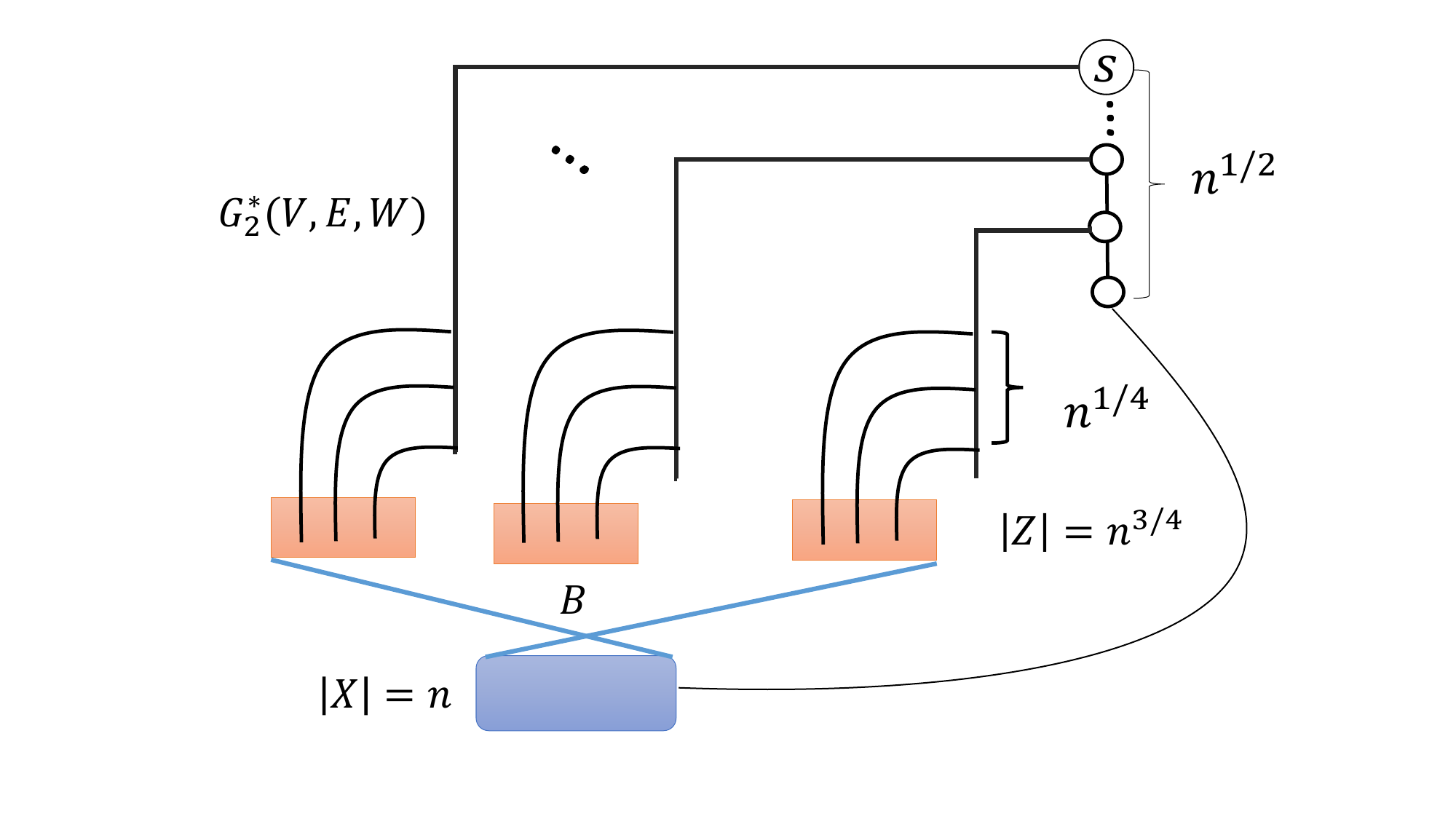}
\caption{\sf Illustration of the lower bound graph $G^*_f(V,E,W)$ for $f=2$. The edge weights of the bipartite graph are monotone increasing as a function of the leaf index from left to right. \label{fig:LB-graph-final}
}
\end{center}
\end{figure}